\documentclass[
 prb,
 amsmath,amssymb,
 reprint,
superscriptaddress
]{revtex4-1}

\usepackage[ruled, vlined]{algorithm2e}
\usepackage{chemformula}
\usepackage[version=3]{mhchem} 
\usepackage{graphicx}
\usepackage{dcolumn}
\usepackage{bm}
\usepackage{color}
\usepackage[normalem]{ulem}
\usepackage{simplewick}
\usepackage{qcircuit}
\usepackage{braket}
\usepackage{float}
\usepackage{multirow}
\usepackage{tikz}
\usepackage[colorlinks=true,urlcolor=blue,citecolor=blue,linkcolor=blue]{hyperref}
\usepackage{stmaryrd}
\usepackage{soul}
\usepackage{import}
\usepackage{booktabs}       
\usepackage{tabularx}

\usepackage{amsmath,amssymb,amsthm}
\usepackage{physics}

\newtheorem{theorem}{Theorem}

\newtheorem{lemma}{Lemma}

\theoremstyle{definition}

\theoremstyle{remark}

\newtheorem{assumption}{Assumption}

\newcommand{\etal}{\emph{et al.}}

\begin{document}

\title{Quantum-Enhanced Neural Exchange-Correlation Functionals}

\date{\today}

\begin{abstract}

Kohn-Sham Density Functional Theory (KS-DFT) provides the exact ground state energy and electron density of a molecule, contingent on the as-yet-unknown universal exchange-correlation (XC) functional. 
Recent research has demonstrated that neural networks can efficiently learn to represent approximations to that functional, offering accurate generalizations to molecules not present during the training process. 
With the latest advancements in quantum-enhanced machine learning (ML), evidence is growing that Quantum Neural Network (QNN) models may offer advantages in ML applications. 
In this work, we explore the use of QNNs for representing XC functionals, enhancing and comparing them to classical ML techniques.
We present QNNs based on differentiable quantum circuits (DQCs) as quantum (hybrid) models for XC in KS-DFT, implemented across various architectures. 
We assess their performance on 1D and 3D systems.
To that end, we expand existing differentiable KS-DFT frameworks and propose strategies for efficient training of such functionals, highlighting the importance of fractional orbital occupation for accurate results. 
We prove that, under specific assumptions, noise-induced errors do not accumulate across SCF iterations, and we confirm this result through numerical experiments.
Our best QNN-based XC functional yields energy profiles of the H$_2$ and planar H$_4$ molecules that deviate by no more than 1 mHa from the reference DMRG and FCI/6-31G results, respectively.
Moreover, they reach chemical precision on a system, H$_2$H$_2$, not present in the training dataset, using only a few variational parameters.
This work lays the foundation for the integration of quantum models in KS-DFT, thereby opening new avenues for expressing XC functionals in a differentiable way and facilitating computations of various properties. 

\end{abstract}

\author{Igor O. Sokolov}
\email{sokolov.igor.ch@gmail.com}
\affiliation{Pasqal, 2 av. Augustin Fresnel, Palaiseau 91120, France}

\author{Gert-Jan Both}
\affiliation{Pasqal, 2 av. Augustin Fresnel, Palaiseau 91120, France}

\author{Art D. Bochevarov}
\affiliation{Schr\"odinger Inc., 1540 Broadway, 24th Floor, New York, NY 10036, United States}

\author{Pavel A. Dub}
\affiliation{Schr\"odinger Inc., 9868 Scranton Rd Suite 3200, San Diego, CA 92121, United States}

\author{Daniel S. Levine}
\affiliation{Schr\"odinger Inc., 1540 Broadway, 24th Floor, New York, NY 10036, United States}

\author{Christopher T. Brown}
\affiliation{Schr\"odinger Inc., 1 Main St 11th Floor, Cambridge, MA 02142, United States}

\author{Shaheen Acheche}
\affiliation{Pasqal, 2 av. Augustin Fresnel, Palaiseau 91120, France}

\author{Panagiotis Kl. Barkoutsos}
\affiliation{Pasqal, 2 av. Augustin Fresnel, Palaiseau 91120, France}

\author{Vincent E. Elfving}
\affiliation{Pasqal, 2 av. Augustin Fresnel, Palaiseau 91120, France}

\maketitle

\section{\label{sec:intro} Introduction}

Density Functional Theory (DFT) has emerged as a major direction of research in modern quantum chemistry, enabling an efficient prediction of the electronic structure and properties of molecules and materials.~\cite{Parr-Yang-DFT, Koch-Holthausen-DFT, Engel-Dreizler-DFT} 
It has become the dominant method for solving the Schr\"{o}dinger equation in physics-based 
atomistic simulations.~\cite{Becke-DFT-JCP-2014, Jones-DFT-RMP-2015} 
DFT's effectiveness largely stems from the ability to express the ground state density of interacting electrons via a system of non-interacting electrons, the so-called Kohn-Sham system.
However, DFT can yield exact ground state energy and electron density only if the exact exchange-correlation (XC) functional is provided.
Despite the widespread applicability of the technique, the development of accurate and versatile DFT XC functionals remains a challenge.~\cite{Cohen-CR-2012, Medvedev-Science-2017, Kepp-Science-2017}  
Consequently, DFT faces problems such as underestimation of reaction barriers, band gaps~\cite{cohen2008insights}, and poor handling of degenerate states potentially due to delocalization and static correlation errors.~\cite{cohen2008fractional, becke2014perspective} 
The universal functional capable of solving the aforementioned problems and capturing the diverse electronic properties of molecules, solids and surfaces remains unknown. 
In its pursuit, hundreds of approximate XC functionals have been devised.~\cite{Mardirossian-MP-2017, Goerigk-zoo-PCCP-2017} 
Many of these XC functionals are tailored to certain types of chemical systems and properties thereof, and often require prior knowledge of the targeted data to ensure the success of the method.~\cite{Caldeweyher-D4-JCP-2019, Verma-TC-2020, Rohrdanz-LR-JCP-2009} 
An additional challenge when constructing such functionals lies in achieving a good balance between accuracy and computational efficiency.~\cite{Mohr-PCCP-2015, Bursch-ACIE-2022}

A number of studies~\cite{tozer1996exchange, tozer1997calculation, Brockherde-NatCom-2017, Nagai-npj-2020, Dick-NatCom-2020, Kalita-ACR-2021, Kirkpatrick-Sci-2021, ksr, Cuierrier-JCP-2022, Riemelmoser-JCTC-2023} propose the use of compelling capabilities of machine learning (ML) techniques as a solution for the efficient design of XC functionals that can be used for generic DFT calculations. 
ML algorithms, ranging from deep neural networks to sophisticated regression models, have demonstrated their potential to unlock hidden patterns in vast datasets of chemical systems, allowing for the estimation of a universal XC functional. 
This combination of quantum chemistry and artificial intelligence heralds a promising era for XC functional design with many demonstrations for different systems~\cite{kasim2021learning, nagai2022machine} indicating the increase of the reach of the DFT methods.
The majority of these ML-based techniques employ differentiable programming as a widely used paradigm in deep learning (DL), which involves optimizing parameters using gradient-based methods during the training process.~\cite{wu2023construct}
In particular, existing chemistry codes are being refactored to incorporate automatic differentiation.~\cite{zhang2022differentiable, casares2023graddft}
Training neural XC functionals within a differentiable DFT procedure was shown to provide superior generalization capabilities in 1D~\cite{ksr, kalita2022well} compared to the previous approach.~\cite{schmidt2019machine}
Extensions to 3D with variation of model inputs (e.g., density and its gradients), locality~\cite{kasim2021learning, casares2023graddft} and novel genetic approaches~\cite{ma2022evolving} followed, to name just a few.
Nevertheless, a universal neural XC functional has not been trained yet.~\cite{wu2023construct}
In an effort to meet this challenge, we propose a novel method for learning representations of XC functionals.

Recently, a novel DL framework within the broader scope of quantum machine learning (QML) has been introduced.~\cite{kyriienko2021solving} 
This framework utilizes differentiable quantum circuits (DQCs), capitalizing on the unique methodologies inherent in quantum computation.%
This category of circuits facilitates the construction of quantum neural networks (QNNs), enabling efficient function representation~\cite{abbas2021power} and derivative computation through analytic differentiation rules.~\cite{mitarai2018quantum, schuld2019evaluating, kyriienko2021generalized}
Various QNN models, which utilize different data encoding methods and distinct quantum circuit designs, have emerged.~\cite{cerezo2021variational}
A relevant quantum advantage in QML has been shown for situations where purely quantum data (quantum state) is used as the model input~\cite{cong2019quantum} and when the quantum model is learning from quantum experiments.~\cite{huang2022quantum}
Given the predominance of classical data in ML, the appeal of QNNs lies in their high effective dimension and rich expressivity,~\cite{abbas2021power} which are determined by the unitary mappings of data to quantum states.~\cite{schuld2021effect} 
However, caution is warranted, as commonly used metrics for analyzing QNNs -- such as expressivity and entangling capability -- may not directly correlate with a model's ability to perform well on a specific classical dataset.
Despite sharing many issues with classical neural networks, such as problems with barren plateaus,~\cite{larocca2022diagnosing} as well as unique issues such as the absence of efficient backpropagation,~\cite{abbas2023quantum} advances have been made in addressing these through specific techniques~\cite{beer2020training, sack2022avoiding} and efficient output extraction methods.~\cite{huang2020predicting, miller2022hardware} 
Moreover, systematic design principles for QNN architectures that are simultaneously trainable and provably non-classically simulable have been proposed~\cite{gil2024relation}.

Many studies have explored the combination of DFT with variational quantum algorithms through auxiliary Hamiltonians,~\cite{senjean2023toward} embedding schemes,~\cite{rossmannek2021quantum} or reduced density matrix methods,~\cite{schade2022parallel} among others. 
However, to the best of our knowledge, the direct application of QML for the functional representation of XC has not been explored. 
Recent developments have further reinforced the connection between DFT and QML.
Gaitan and Nori~\cite{gaitan2009density} established a mapping between qubit systems and fermionic lattice models, providing a foundation for employing ground-state and time-dependent DFT to study properties of large-scale quantum systems, including the determination of spectral gaps.
On the other hand, Baker and Poulin~\cite{baker2020density} caution against overly optimistic expectations for QML, demonstrating that no QML algorithm can discover the universal functional in polynomial time unless QMA-complete problems reduce to the complexity class BQP.
Within the probably approximately correct (PAC) framework of machine learning, QML can at best provide a polynomial reduction in oracle queries (training points) when learning the XC functional.
We take an optimistic perspective, analogous to the situation in classical ML where theoretical arguments once predicted that training deep-learning models (e.g., due to numerous local minima) would not succeed.
Nonetheless, the remarkable success of large language models has demonstrated that practical performance can exceed pessimistic theoretical limitations.
Therefore, in this work, we investigate the synergy between QML and DFT for designing XC functionals.~\cite{patent}
As suggested in ref.~\cite{schuld2022quantum}, our methodology incorporates quantum models as part of larger classical ML pipelines.
The goal is to use techniques arising from the QML domain to process empirical data from DFT (or general quantum chemistry) calculations within the supervised ML approach to create more robust and versatile XC functionals that can be applied to a wider range of problems in the quantum chemistry domain.

Our main contributions are:
\begin{itemize}
\item[(\textbf{i})] Expanding of KS-DFT frameworks for quantum XC functional training and implementing in an open-source package QEX~\cite{qex2025pasqal} written in JAX~\cite{jax2018github}.
\item[(\textbf{ii})] Demonstrating the importance of fractional occupation in 3D KS-DFT for improving accuracy.
\item[(\textbf{iii})] Introducing QNNs models for XC representation, showcasing efficiency in 1D and 3D systems.
\item[(\textbf{iv})] Showcasing that quantum models are at least on par with classical models in terms of accuracy and parameter-efficiency on the tested molecular systems.
\item[(\textbf{v})] Providing error bounds for the KS-DFT self-consistent loop under reasonable assumptions in noisy settings, e.g., when using noisy QNN-based XC functionals.
\end{itemize}
This work is organized as follows: 
In Sec.~\ref{sec:theory}, we give the theoretical background of QML for KS-DFT and define our models for XC functionals.
In Sec.~\ref{sec:methods}, we provide details of our framework implementation.
In Sec.~\ref{sec:results}, we present our empirical study of the proposed architectures, analyse the effect of quantum and classical layers within the same backbone architecture, discuss the required quantum resources and error bounds when considering typical hardware noise.
Finally, in Sec.~\ref{sec:conclusion}, we discuss our conclusions and give an outlook on the development of useful quantum-enhanced XC functionals.

\section{\label{sec:theory} Theory}

In this section, we introduce the KS-DFT approach, which serves as the foundation for training and evaluating our quantum machine-learned XC functionals.
We then delve into the utilization of quantum ML for training XC functionals.
Subsequently, we discuss how neural XC functionals are embedded in KS-DFT and explore their architectural details.

\subsection{Kohn-Sham Density Functional Theory}

Contemporary KS-DFT implementations~\cite{Koch-Holthausen-DFT} achieve the electronic ground state through a solution of the KS equations
\begin{equation}
\hat H[n](\mathbf{r}) \phi_i(\mathbf{r}) =\epsilon_i \phi_i(\mathbf{r}),
\label{eq:ks-eq}
\end{equation}
to self-consistency with respect to the electronic density $n(\mathbf{r})$.
The Hamiltonian operator is given by $\hat H[n](\mathbf{r})=-\nabla^2 / 2+v_\mathrm{KS}[n](\mathbf{r})
$
as a function of the electron density $n(\mathbf{r})$.
The latter is defined as the sum of probability densities over all $N_{\mathrm{occ}}$ occupied orbitals, 
$n(\mathbf{r}) =\sum^{N_{\mathrm{occ}}}_{i=1} \left|\phi_i(\mathbf{r})\right|^2,
$
where $\phi_i$ is a one-electron KS orbital with corresponding energy $\epsilon_i$.
The KS potential is given by $ v_{\mathrm{KS}}[n](\mathbf{r}) = v(\mathbf{r})+v_\mathrm{H}[n](\mathbf{r})+v_{\mathrm{XC}}[n](\mathbf{r}),$
which contains the external ionic Coulomb potentials $v(\mathbf{r})$, the Hartree potential $v_{\mathrm{H}}[n](\mathbf{r})$ , and the XC potential $v_{\mathrm{XC}}[n](\mathbf{r})$. The XC potential is the functional derivative of the XC energy, $E_{\mathrm{XC}}$, with respect to the electron density, i.e. $v_{\mathrm{XC}}[n](\mathbf{r})=$ $\delta E_{\mathrm{XC}}[n] / \delta n(\mathbf{r})$.
The XC energy $E_{\mathrm{XC}}$ is defined as 
\begin{equation}
E_{\mathrm{XC}}[n]=\int \mathrm{d} \mathbf{r} \, \epsilon_{\mathrm{XC}}[n](\mathbf{r}) n(\mathbf{r}),
\label{eq:e-xc}
\end{equation}
where $\epsilon_{\mathrm{XC}}[n](\mathbf{r})$ is the XC energy per electron. 
Finally, the total energy is obtained as the sum, 
$E[n]=T_{\mathrm{S}}[n]+V[n]+E_{\mathrm{H}}[n]+E_{\mathrm{XC}}[n]$,
of the non-interacting kinetic energy $T_{\mathrm{S}}[n]$, the external potential energy $V[n]$, the Hartree energy $E_{\mathrm{H}}[n]$ and the XC energy $E_{\mathrm{XC}}[n]$.

In practice, given a finite basis set $\{\xi_i\}_{i=1}^{N_b}$, a linear combination of $N_b$ basis functions, $\phi_i(\mathbf{r})= \sum^{N_b}_{j=1} p^{j}_{i} \xi_j(\mathbf{r})$, is optimized to produce suitable orbitals $\phi$.
This involves the solution of Roothan's equation defined as 
\begin{equation}
\mathcal{F}[n] \mathbf{p}_i=\epsilon_i S \mathbf{p}_i,
\label{eq:roothan}
\end{equation}
where $\mathbf{p}_i = (p^{1}_i,...,p^{N_b}_i)^{T}$, $S_{i j}=\int \mathrm{d} \mathbf{r} \xi_i(\mathbf{r}) \xi_j^*(\mathbf{r})$ is the overlap matrix and $\mathcal{F}[n]$ is the Fock matrix as a functional of the electron density profile $n$, with elements $\mathcal{F}_{i j}=\int \mathrm{d} \mathbf{r}\, \xi_i(\mathbf{r}) \hat{H}[n](\mathbf{r}) \xi_j^*(\mathbf{r}).$
Eq.~\eqref{eq:roothan} is then solved, e.g., numerically by an eigendecomposition of the Fock matrix to obtain the orbital coefficients $\{ \mathbf{p}_i\}$ that correspond to some density profile $n$. 
Self-consistency is achieved when the electron density no longer changes within a specified tolerance over consecutive $N_{\mathrm{KS}}$ KS iterations (\textit{i.e.}, solutions of Eq.\eqref{eq:roothan}).
The resulting energy $E$ can then be compared to validated experimental data, with the quality of the result dependent on the choice of the (neural) functional used to represent $E_{\mathrm{XC}}[n]$.
If the exact XC functional was known, Hohenberg-Kohn theorems~\cite{ks-theorem} guarantee that the self-consistent solution yields the exact ground-state energy and density.

\subsection{Embedding of neural XC functionals within KS-DFT \label{sec:inc-xc}}

In Fig.~\ref{fig:locality}, we summarize the approaches to incorporate our neural XC models of different localities in KS-DFT.
In particular, within the established Kohn-Sham framework, neural XC models are integrated using Eq.~\eqref{eq:e-xc}.
These models are represented by differentiable functions $f_{\theta}: (n, \mathbf{r}) \longrightarrow \epsilon_{\mathrm{XC}}(\mathbf{r})$ with parameters denoted by $\theta$.
Consequently, the machine-learned XC energy can be expressed in general as
$E_{\mathrm{XC}}[n] \approx E^{\theta}_{\mathrm{XC}}[n]=\int d \mathbf{r} f_{\theta}(n, \mathbf{r}) n(\mathbf{r}),$
where $E^{\theta}_{\mathrm{XC}}[n]$ approximates the exact $E_{\mathrm{XC}}[n]$.
In practice, an integration grid is used to compute the integral on a finite educated choice of $N_{\mathrm{grid}}$ grid points.
The XC energy is integrated as
\begin{equation}
E_{\mathrm{XC}}[n] \approx E^{\theta, 1}_{\mathrm{XC}}[n]=\sum^{N_{\mathrm{grid}}}_{i=1} w_i f_{\theta}(n, \mathbf{r}_i) n(\mathbf{r}_i),
\label{eq:ksr-xc}
\end{equation}
where the grid weights $\{w_i\}^{N_{\mathrm{grid}}}_{i=1}$ are computed for each grid type and chemical system. The values $\{\mathbf{r}_i\}^{N_{\mathrm{grid}}}_{i=1}$ denote the Cartesian coordinates associated with grid points.
We call this scheme the \textit{local energy embedding} since local XC energy density $\epsilon_{\mathrm{XC}}(\mathbf{r})$ is required (see Fig.~\ref{fig:framework-detailed}(a)).

\begin{figure}[ht!]
    \centering
    \includegraphics[width=1\columnwidth]{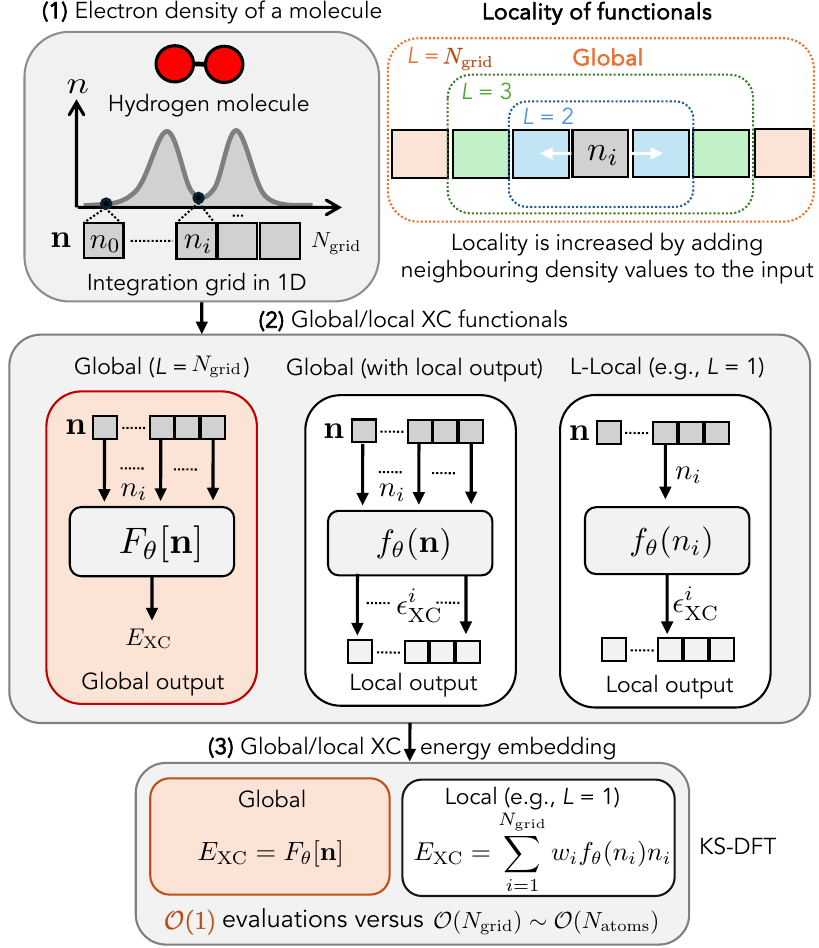}
    \caption{
    (1) Electron density of a molecule is expressed on an integration grid. 
    (2) The neural XC functional takes $L$ neighbouring values of density as input.
    (3) The output of such functional can be interpreted as local or total XC energy, with respective embeddings in KS-DFT (Eqs.~\eqref{eq:ksr-xc} and~\eqref{eq:gl-xc}).
    With the global approach, only $\mathcal{O}(1)$ calls to the functional is necessary to obtain XC energy instead of $\mathcal{O}(N_{\mathrm{grid}})$.
    }
    \label{fig:locality}
\end{figure}

Another approach is to represent the total XC directly with a neural functional $F_{\theta}[n]$, bypassing the integration as
\begin{equation}
E_{\mathrm{XC}}[n] \approx E^{\theta, 2}_{\mathrm{XC}}[n]= F_{\theta}[n],
\label{eq:gl-xc}
\end{equation}
where $F_{\theta}[n]$ learns implicitly the relation between the XC energy and the grid during training through
$\{n_i\}^{N_{\mathrm{grid}}}_{i=1}$ input density values.
Both equations, Eqs.~\eqref{eq:ksr-xc} and \eqref{eq:gl-xc}, provide distinct ways of embedding neural networks as XC functionals in KS-DFT frameworks.
We call this approach the \textit{global energy embedding} since the model requires the whole density profile to infer the total XC energy (see Fig.~\ref{fig:framework-detailed}(a)).
Note that despite avoiding grid integration in the global energy embedding, both techniques require the computation of the local XC potential $v_{\mathrm{XC}}[n](\mathbf{r})$ at every grid point.

Neural XC functionals are also characterized by various degrees of \textit{locality} that represents the number of density points provided as inputs.~\cite{xc-mlp-1d}
Locality was shown to greatly impact the ability to learn accurate XC functionals.~\cite{ksr}

In particular, \textit{$L$-local} ($L\in \llbracket 1, N_{grid}\rrbracket$) models are represented by the functions 
\begin{equation}
 f^{\mathrm{\mathrm{local}}}_{\theta}:\mathbb{R}^{L} \to \mathbb{R},\ \{n\}^{L} \rightarrow \epsilon^{\theta}_{\mathrm{XC, i}},
 \label{eq:loc}
\end{equation}
where the model function has access only to some $L$ neighboring density points (kernel size) at a time and is applied in succession on, e.g., $\lceil N_{\mathrm{grid}} / L \rceil$ points.~\cite{xc-mlp-1d} In the limit where $L=1$, the $1$-local model is represented by the functions $f^{local}_\theta: n_i \to \epsilon^\theta_{\text{XC}, i}$. Local models are embedded in DFT via Eq.~\eqref{eq:ksr-xc}.

\textit{Global} models are defined as a functional of the whole density,
\begin{equation}
 F^{\mathrm{global}}_{\theta}:\mathbb{R}^{N_{grid}} \to \mathbb{R},\ \{n\}_{i=1}^{N_{\mathrm{grid}}} \rightarrow E^{\theta}_{\mathrm{XC}},
 \label{eq:glob1}
\end{equation}
where the NN represents the whole integrated XC energy and is embedded in DFT via Eq.~\eqref{eq:gl-xc}.
Alternatively, global models can also be embedded in DFT via Eq.~\eqref{eq:ksr-xc} as
\begin{equation}
 f^{\mathrm{global}}_{\theta}:\mathbb{R}^{N_{grid}} \to \mathbb{R}^{N_{grid}},\  \{n\}_{i=1}^{N_{\mathrm{grid}}} \rightarrow \{ \epsilon^{\theta}_{\mathrm{XC, i}} \}_{i=1}^{N_{\mathrm{grid}}},
 \label{eq:glob2}
\end{equation}
where local XC energy density depends on the whole electron density.
Next, we present the concrete realisation of our XC models in classical and quantum settings.

\subsection{Architectures of neural XC functionals \label{sec:models}}

In this section, we propose a blueprint of basic model architectures that can couple classical and quantum layers.
Neural XC functionals are typically represented as compositions of differentiable functions, implemented through stacks of neural network (NN) layers.
\begin{equation}
    f_{\theta}(n, \mathbf{r}) = (f_{\theta_{N_l}}^{N_l} \circ f_{\theta_{N_l-1}}^{N_l-1}\circ ... \circ f^{1}_{\theta_{1}})(n, \mathbf{r}),
    \label{eq:layers}
\end{equation}
with parameters $\theta = (\theta_{N_l}, ..., \theta_{1})$ and $N_l$ layers.

\subsubsection{Classical architectures}

The most general neural network is a Multilayer Perceptron (MLP) which can approximate any function given enough layers~\cite{hornik1989multilayer}.
Since we are exploring the most general quantum models, MLPs are suitable candidates for comparative studies.

For our purpose, we implement MLPs with $N_{i}$ inputs, $N_l$ layers, and $N_{\mathrm{n}}$ neurons per layer. We introduce two MLP variants: Local MLP (LMLP) uses local density information at $N_{\mathrm{grid}}$ points to map:
\begin{equation}
f^{\mathrm{LMLP}}_{\theta}: n_i \rightarrow \epsilon^{\theta}_{\mathrm{XC}, i}\label{eq:lmlp}
\end{equation}
while Global MLP (GMLP) is a functional of the whole density, mapping:
\begin{equation}
     F^{\mathrm{GMLP}}_{\theta}: \{n\}_{i=1}^{N_{\mathrm{grid}}} \rightarrow E^{\theta}_{\mathrm{XC}}.
\end{equation}

Finally, another notable neural network is the Global Kohn-Sham Regularizer (GKSR) model, which combines local and global convolutional layers. GKSR incorporates inductive biases such as a self-interaction gate and a negativity transform, addressing common issues like self-interaction errors and imposing the negativity condition on the XC energy.~\cite{ksr}
In our implementation, the GKSR model is represented by the function:
\begin{equation}
 f^{\mathrm{GKSR}}_{\theta}: \{n\}_{i=1}^{N_{\mathrm{grid}}} \rightarrow \{ \epsilon^{\theta}_{\mathrm{XC, i}} \}_{i=1}^{N_{\mathrm{grid}}}
 \label{eq:}
\end{equation}.

In the following section, we discuss how XC functionals can be represented as quantum circuits.

\subsubsection{Quantum architectures}

A widely adopted strategy for QML models involves the utilization of a QNN, typically constructed using differentiable quantum circuits.~\cite{kyriienko2021solving}
In this work, our quantum XC functionals are based on such QNNs, which consist of three blocks:  a feature map ${{U}}_{f}$, an ansatz ${{U}}_{a}$ and a measurement of an operator ${\hat{O}}$.
Typically, a  $N_\mathrm{q}$-qubit quantum computer is initialized at some initial state $| 0 \rangle$. 
Here we consider the zero computational basis state on all qubits.
A feature map is used to embed the data (e.g., electron density $n$) in the Hilbert space, ${{U}}_{f}(n)|0\rangle$ in a differentiable manner for the computation of XC potential $v_{\mathrm{XC}}$.
The discretized density $\{n_i\}_{i=1}^{N_{\mathrm{grid}}}$ values are embedded in the feature map as gate parameters.
Then, an ansatz circuit ${{U}}_{a}({\theta_{a}})$ is used to act on the embedded data with local qubit operations and entangling gates that mediate interactions to represent the most general action on the data. 
Finally, the output of the model is extracted via projective measurements (shots) of an arbitrary cost operator $\hat{C} = \sum_{i=1}^{N_\mathrm{o}} \theta^{i}_c \hat{O}_i$, where $\theta_{c}$ are fixed or trainable parameters associated with Pauli operators $\hat{O} \in \{I, \hat X, \hat Y, \hat Z \}^{\otimes N_\mathrm{q}}$.
We assume the parameters $\theta_{c}$ are non-trainable, unless specified otherwise.
The cost operator is chosen to be the total magnetization, composed of local observables,
    $\hat{C} = \sum_{i=1}^{N_\mathrm{q}} \hat{Z_i},$
which was shown to exhibit, at worst, polynomially vanishing gradients if the depth scales as $\mathcal{O}(\text{log}(N_q))$, thereby mitigating cost-function-induced barren plateaus~\cite{cerezo2021cost}.
The representation of such model is given by a function 
\begin{equation}
f^{\mathrm{q}}_{\theta}(n)=\langle 0 |U^{\dagger}(n, {\theta}) \hat{C} U(n, {\theta})|0 \rangle,
\label{eq:qmodel}
\end{equation}
where 
    $U(n, {\theta}) = \prod^{N_l}_{l=1} \left( \prod^{N_d}_{d=1} {{U}}^{l,d}_{a}({\theta^{l,d}_{a}}){{U}}^{l}_{f}(n)\right)$
is the general representation of the model with $N_d$ ansatz layers and $N_l$ total repetition of ansatz and feature map, which implements the data re-uploading strategy.~\cite{perez2020data}
The repetitions of layers are akin to the number of layers of classical NNs and augment the expressivity of the model.
Given enough layers, classical and quantum models can, in principle, approximate a given function.~\cite{hornik1989multilayer, perez2020data, perez2021one}
In Appendix~\ref{app:qnn}, we give details of typical building blocks of QNNs.

To parallel the local MLP, we formulate a fully Local QNN (LQNN) designed to take a single density point $n_i$, from density vector $\mathbf{n}$, as input and yield the expectation of an observable:
\begin{equation}
 f^{\mathrm{LQNN}}_{\theta}: n_i \rightarrow \epsilon^{\theta}_{\mathrm{XC, i}}= \langle 0 |U^{\dagger}(n_i, {\theta}) \hat{C} U(n_i, {\theta})|0 \rangle,
 \label{eq:lqnn}
\end{equation}
defined through Eq.~\eqref{eq:loc} and integrated into DFT via Eq.~\eqref{eq:ksr-xc}.
We explore local QNNs with both product (Eq.~\eqref{eq:fm-prod}) and Chebyshev (Eq.~\eqref{eq:fm-cheb}) feature maps, named Local PrQNN and Local ChQNN, respectively.
The ansatz adheres exactly to the structure specified in Eq.~\eqref{eq:state}.
Note that for this model it is possible to consider $L$-local case,
$f^{\mathrm{LQNN}}_{\theta}: \{n_i\}^{L} \rightarrow \epsilon^{\theta}_{\mathrm{XC, i}}= \langle 0 |U^{\dagger}(n, {\theta}) \hat{C} U(n, {\theta})|0 \rangle,$ where a collection of $L$ neighboring densities is embedded in a single feature map. 
However, as shown in ref.~\onlinecite{ksr}, only global functional (e.g., $L=N_{\mathrm{grid}}$) is capable of best approximating exact XC functionals.
Hence, in the next section, we develop model architectures that embed the whole density.

\subsubsection{Feature maps for quantum models}\label{sec:feature_maps}

The choice of feature map is crucial for the expressiveness of quantum XC functionals. We consider encoding data in circuit parameters and amplitudes using general angle embeddings as specific cases of feature maps:
\begin{equation}
U_{f}(\mathbf{n}) = \bigotimes_{k=1}^{N_q} e^{-\frac{i}{2} \varphi(n_k) \hat{G}_k}\label{eq:featuremap},
\end{equation}
where $\hat{G}$ are generators (possibly multi-qubit operators), and $\varphi: \mathbb{R} \rightarrow \mathbb{R}$ is a data transformation function. We explore three established strategies for embedding data in quantum circuits. 
The choice of those strategies can be encouraged by viewing the feature maps selected for this work as Fourier series, which allows for a deeper understanding of their structure and behavior. 
According to Schuld~\etal,~\cite{schuld2021effect} such feature maps admit a Fourier series representation with a (potentially) exponential number of modes in the number of encoding gates or layers. 
This means that QNNs have an expressive power that may grow exponentially. 
It remains an open question whether the greater expressivity of QNNs provides a tangible advantage in representing XC functionals, even when such models cannot be reliably approximated classically for large $N_q$~\cite{landman2022classically}.
However, despite decades of research towards a universal XC functional, no single functional has proven accurate across all molecular systems.
In this context, such new model types could offer an intriguing alternative.

\subsubsection*{Product features map}

We explore the $1$-local product features map that embeds the data value $n_k$ in the rotation angles of $R_y$ gates of all qubits:
\begin{equation}
    U^{(1)}_{f,\mathrm{pr}}(n_k) = \bigotimes_{l=1}^{N_q} e^{-\frac{i}{2} n_k \hat{Y}_l}\label{eq:fm-prod}.
\end{equation}
Similarly to $L$-local models discussed earlier, an $L$-local product feature map can encode multiple data points within a single feature map. Nevertheless, this approach was not employed in the present work.

\subsubsection*{$1$-local Chebyshev feature map}
Chebyshev feature map utilizes a polynomial basis\cite{schuld2021effect}:
\begin{equation}
 U^{(1)}_{f,\mathrm{ch}}(\mathbf{n}) = \bigotimes_{k=1}^{N_q} e^{- i \cdot k\cdot\mathrm{arccos}(n_k) \hat{Y}_k}.\label{eq:fm-cheb}
\end{equation}
However, normalization of the density to a unit vecteor, \textit{i.e.} $\mathbf{n} \leftarrow \mathbf{n}/ || \mathbf{n} ||$, is necessary. 
It suppresses electron count information and thus renders this approach unsuitable for density-based models.

\subsubsection*{Amplitude feature map}
Finally, the amplitude feature map maps the density vector to a quantum state as $|\mathbf{n}\rangle = U^{(N_{\text{grid}})}_{f,\mathrm{a}}(\mathbf{n}) |0\rangle$. 
However, this method presents two important drawbacks: the density is also normalized to a unit vector which suppresses the information of the number of electrons and, in general, an efficient implementation algorithm for the unitary $U^{(N_{\text{grid}})}_{f,\mathrm{a}}(\mathbf{n})$ does not exist for arbitrary vectors $\mathbf{n}$~\cite{nakaji2022approximate}.

\subsubsection{Quantum-classical hybrid architectures}

We propose to leverage advantages of classical NNs and combine them with quantum models.
Notably, the number of elements $N_{\mathrm{grid}}$ in the electron density $\mathbf{n}$ scales linearly with the number of atoms.~\cite{dasgupta2017standard}
Even for small diatomic molecules, there are hundreds of grid points for the smallest grid density available with the best integration grids~\cite{dasgupta2017standard}. 
This poses an issue with the feature maps based on angle encoding since the number of qubits $N_{\mathrm{q}}$ required to embed the entire density \(\mathbf{n}\) equals its size $N_{\mathrm{grid}}$, posing severe requirements for potential implementation on near-term quantum computers.
Hence, we devise a model, Global QNN (GQNN), that is composed of two layers:
\begin{equation}
f^{\mathrm{GQNN}}_{\theta}(n) = (f_{\theta_{\mathrm{c}}}^{\mathrm{c}} \circ  f^{\mathrm{q}}_{\theta_{\mathrm{q}}})(n),
\label{eq:globalqnn}
\end{equation}
where $f^{\mathrm{c}}$ and $f^{\mathrm{q}}$ represent the classical and quantum layers with associated parameters $\theta_{\mathrm{c}}$ and $\theta_{\mathrm{q}}$, respectively. 
In principle, any differentiable classical layer can be used to map the outputs of $f^{\mathrm{q}}_{\theta_{\mathrm{q}}}$ to XC energy.
Here, the quantum layer consists of $L$-local QNNs defined as
\begin{equation}
 f^{\mathrm{q}}_{\theta}: \{n_i\}^{L} \rightarrow \epsilon_{\mathrm{XC}, i} = \langle 0 |U^{\dagger}(n, {\theta}) \hat{C} U(n, {\theta})|0 \rangle,
 \label{eq:fq}
\end{equation}
with a product feature map and an alternate ladder ansatz.
We first partition the input $\mathbf{n}$ in $N_{\mathrm{b}} = \lceil N_{\mathrm{grid}/} / L \rceil$ batches and then apply $f^{\mathrm{q}}_{\theta}$ to each batch.
This results in $N_{\mathrm{b}}$ outputs that become the inputs of a classical function $f_{\mathrm{c}}$ defined as
\begin{equation}
 f^{\mathrm{c}}_{\theta_{\mathrm{c}}}: \{\epsilon_{\mathrm{XC}, i}\}_{i=1}^{N_{\mathrm{b}}} \rightarrow E^{\theta}_{\mathrm{XC}},
 \label{eq:fc}
\end{equation}
where, for simplicity, the classical layer is the standard MLP with $N_{\mathrm{b}}$ input neurons and one output neuron (potentially with additional MLP layers). 
In general, any classical model architecture can be used to represent the classical layer.

\subsubsection{Quantum-inspired architectures}

Another possibility for designing XC functionals is to use simulated quantum models as quantum-inspired (Qi) solution.
We introduce the Global QiNN, which incorporates density through amplitude embedding (Eq.~\eqref{eq:featuremap}) and utilizes the standard ansatz outlined in earlier sections. 
It is represented by the function
\begin{equation}
 F^{\mathrm{QiNN}}_{\theta}: (\{n\}_{i=1}^{N_{\mathrm{grid}}}) \rightarrow E^{\theta}_{\mathrm{XC}},
\end{equation}
where amplitude embedding feature map is used.

Finally, we propose Global Quantum-inspired Convolutional NN (GQiCNN) architecture that implements the composition of quantum layers with angle embedding as
\begin{equation}
f^{\mathrm{GQiCNN}}_{\theta}(n) = (f_{\theta_{c}}^{c} \circ f_{\theta_{N_l}}^{q, N_l} \circ ... \circ  f^{q, 1}_{\theta_{q, 1}})(n),
\label{eq:qicqnn}
\end{equation}
where the product feature map (Eq.~\eqref{eq:fm-prod}) is throughout all the layers. 
We do not employ Chebyshev embedding (Eq.~\eqref{eq:fm-cheb}) due to detrimental normalization of the data required, which suppresses the information about the number of electrons in our architectures.
Similarly to the hybrid architecture (GQNN), we consider quantum layers that consist of $L$-local QNNs, defined in Eq.~\eqref{eq:fq},
with a product feature map and alternate ladder ansatz defined in Appendix~\ref{app:qnn}.
The input is partitioned in $\mathbf{n}$ in $N_{\mathrm{b}} = \lceil N_{\mathrm{grid}/} / L \rceil$ batches and then $f^{\mathrm{q}}_{\theta}$ is applied to each batch.
This is repeated $N_l$ times until the output is a single scalar, which represents $E_{\mathrm{XC}}$.
If not (i.e., odd number of inputs), the output of quantum layers is combined with an MLP as in Eq.~\eqref{eq:fc} to output $E_{\mathrm{XC}}$.
Note that, although all models in this work are simulated, we classify this subset (QiNN and QiCNN) as quantum-inspired because the amplitude encoding used in the Global QiNN (i.e., $|\mathbf{n}\rangle = \sum_{j=1}^{N_q} \mathbf{n}_j |j\rangle$, where $\mathbf{n}$ is assumed to be normalized) typically requires an exponential number of gates~\cite{zhang2021low, nakaji2022approximate}, rendering it impractical on current NISQ devices. 
The QiCNN model is implemented as a sequence of QNNs in which the inputs to each layer are the expectation values from the previous layer, making the model more susceptible to error accumulation, unsuited for current noisy architectures.

\section{\label{sec:methods} Methods}

In this section, we discuss the implementation of our approach in differentiable chemistry frameworks and strategies for efficient training of neural XC functionals.

\subsection{Implementation within differentiable chemistry frameworks \label{sec:framework}}

To evaluate the capabilities of quantum computers for machine learning XC functionals, we based our implementation on two differentiable frameworks.

Li \etal~proposed an important method to facilitate the generalization of neural XC functionals: using the KS iterative procedure as a regularizer.~\cite{ksr}
The novel idea is to compute the energy loss at every KS iteration and the electron density loss at the last KS iteration, and then to differentiate through $N_{\mathrm{KS}}$ iterations (solutions of Eq.~\eqref{eq:ks-eq}) which would enable the model to converge faster to the reference, i.e., in fewer KS iterations, while regularizing its parameters.
To allow this feature, a fully differentiable chemistry code in 1D in JAX~\cite{jax2018github} was developed~\cite{ksr} with the aim of demonstrating the regularization capacity of the framework.~\cite{ksr}
Since it reduces the number of dimensions (to 1D) and hence simplifies the problem, this is an ideal setup for investigation of the capacity of QNN-based models to represent XC functionals.  
We adopt this exact framework with minimal changes to hyperparameters for all our preliminary results in 1D (i.e., only linear molecules). 

However, realistic chemistry simulations require a full 3D treatment of molecular systems. 
To address this demand, Kasim \etal~proposed a 3D differentiable chemistry framework implemented entirely in PyTorch~\cite{paszke2017automatic}, which enabled accurate training of functionals for various diatomic systems.~\cite{kasim2021learning}
However, both specialised frameworks of Li \etal~and Kasim \etal~are not used by computational chemists on a regular basis.
Therefore, in this work, one of our goals is to contribute to an existing effort that focuses on making a popular chemistry framework, PySCF~\cite{sun2018pyscf}, differentiable. 
This differentiable framework is known as PySCF with Automatic Differentiation (PySCFAD)~\cite{zhang2022differentiable}. 
It is similar to Li's 1D framework, as it is also based on JAX, but offers the capability for an accurate treatment of chemical systems in 3D, making it suitable for a wide range of chemical simulations.
Working within the scope of the PySCFAD framework, we replaced the calls to the standard XC functionals in differentiable SCF calculations used in KS-DFT (see Sec.~\ref{sec:theory}) by (quantum) NNs.
This entails that we introduce necessary modifications to the already well-tested KS-DFT implementation.

To facilitate the training of neural XC functionals within KS-DFT in 3D, the spin components of the electron density can be constrained to be equal, i.e., $n_{\uparrow}(\mathbf{r}) = n_{\downarrow}(\mathbf{r})$ and $n(\mathbf{r}) = n_{\uparrow}(\mathbf{r}) + n_{\downarrow}(\mathbf{r})$, which defines the so-called spin-restricted case of KS-DFT (RKS-DFT).

In addition, standard RKS-DFT implementations fail to consistently represent the ground state density, preventing the training of accurate XC models.
Although the exact interacting ground-state wave function cannot, in general, be represented by a single-determinant wave function of a non-interacting KS system with binary orbital occupations (0 or 1), this is not a problem for the density itself. 
Still, RKS-DFT can fail numerically for open-shell systems (typically strongly correlated systems near dissociation), because the usual expression for the density as a sum over occupied orbitals becomes ill-defined: KS valence orbitals may be (near-)degenerate, which can lead in practice to convergence issues in (R)KS-DFT. This, in turn, necessitates fractional occupation of the KS system, which can no longer be described by a single Slater determinant but instead must be treated as an ensemble.
To address this limitation, especially evident in the challenging strongly-correlated regime, we adopt the algorithm of Chai.~\cite{frac-occ}
Other approaches can be employed with an increased computational cost.~\cite{cances2003quadratically}
Our choice was dictated by the simplicity of its integration in existing codes, with detailed proofs given in ref.~\onlinecite{frac-occ} and our implementation in Appendix~\ref{app:frac-occ}.

In summary, we extended two chemistry codes to incorporate QNN-based functionals within KS-DFT:
\begin{enumerate}
    \item For 1D cases, we utilized the KS-DFT implementation by Li~\etal~\cite{ksr}
    \item For 3D cases, we employed the differentiable KS-DFT implementation from PySCFAD~\cite{zhang2022differentiable}, to which we added the FO method developed by Chai~\cite{frac-occ}.
\end{enumerate}
The 1D framework simplifies the training of neural functionals by limiting spatial degrees of freedom, although it restricts the study to linear systems. 
This setup provides an ideal environment to explore quantum models in a simplified context. 
Conversely, the 3D framework permits the evaluation of similar systems within a comprehensive representation, allowing for a realistic analysis.
To provide a quick overview, Figs.~\ref{fig:master} and~\ref{fig:framework-detailed}(a) summarize our framework in the following steps:
\begin{enumerate}
    \item[(1)]  \textit{Generate dataset.} Electronic structure calculations are performed via the best available level of theory (e.g., Full Configuration Interaction (FCI)~\cite{helgaker2013molecular}, DMRG) within a basis set (e.g., 6-31G* or larger) to create a dataset consisting of a set of energies and densities $\{ (E^{i}_{\mathrm{ref}}, \mathbf{n}^{i}_{\mathrm{ref}}) \}_{i=1}^{N_{\mathrm{mol}}}$, with $\mathbf{n}^{i}_{\mathrm{ref}}$ evaluated on a selected grid of size $N_{\mathrm{grid}}$, for $N_{\mathrm{mol}}$ molecular geometries. 
    
    \item[(2)] \textit{Train models.} At every epoch, $N_{\mathrm{KS}}$ iterations are performed,  each calling a neural functional $f_{\theta}$ that outputs the XC energy $E_{\mathrm{XC}}$ via local (Eq.~\eqref{eq:ksr-xc}) or global embeddings (Eq.~\eqref{eq:gl-xc}).   
    The corresponding $v_{\mathrm{KS}}$ potential is obtained via (automatic) differentiation.
    Then, the Fock equations are solved to obtain new orbitals.
    This whole procedure continues until self-consistency. 
    Each KS iteration contributes to the loss $\mathcal{L}$, resulting in regularization.~\cite{ksr}
    At each epoch, trainable parameters $\theta$ are updated by a classical optimizer. 
    
    \item[(3)] \textit{Select the best model.} The converged parameters $\theta^{*}$ define the trained functional $f_{\theta^{*}}$ and are selected at one out of $N_{\mathrm{epochs}}$ epochs that provide the lowest loss on the validation set (i.e., selected geometries that are not present in the train and test sets).
    \item[(4)] \textit{Predict properties.} The selected XC functional $f_{\theta^{*}}$ can be used to evaluate molecular properties. The neural functional can be completely agnostic to the input size ($N_{\mathrm{grid}}$) by design, and hence used for any molecular system.
\end{enumerate}
Our implementation is available on GitHub as part of the QEX software package~\cite{qex2025pasqal}.
In the next section, we describe each step of the framework in detail.

\begin{figure}[ht!]
    \centering
    \includegraphics[width=0.8\columnwidth]{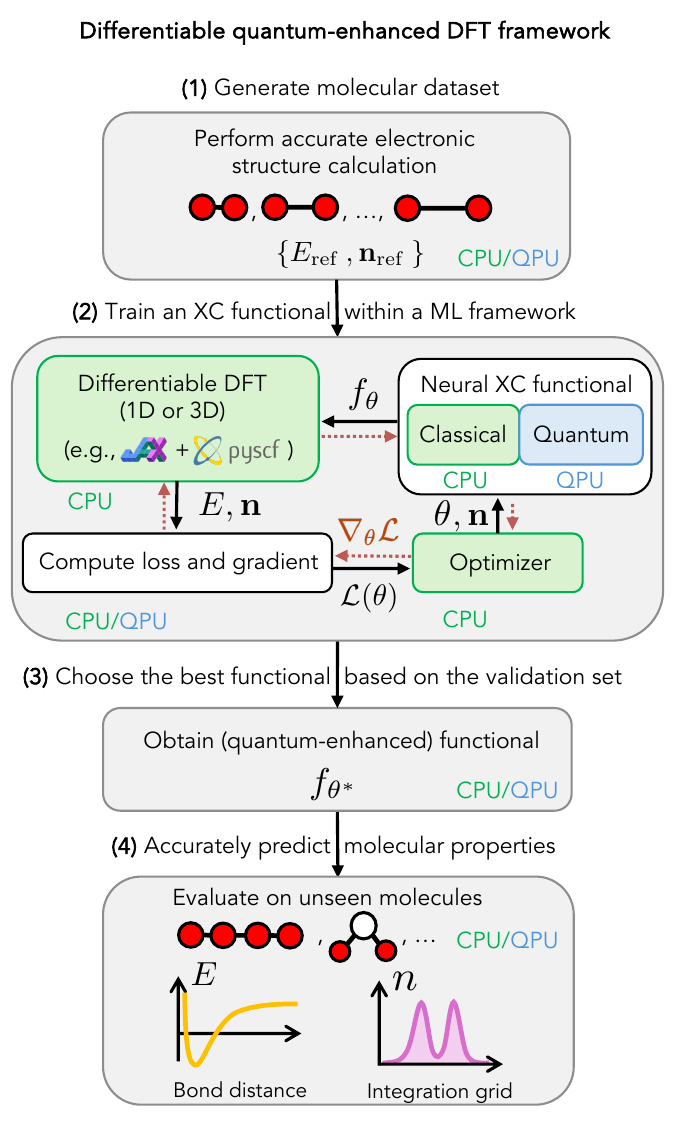}
    \caption{
    Differentiable chemistry framework for training quantum-enhanced XC functionals. 
    (1) The molecular dataset consists of reference energies $E_{\mathrm{ref}}$ and densities $\mathbf{n}_{\mathrm{ref}}$ generated by an accurate \textit{ab initio} method. 
    (2) A differentiable (Q)NN $f_{\theta}$ outputs the XC energy. 
    The corresponding potential is obtained via (automatic) differentiation.
    The Fock equations in KS-DFT are iteratively solved to obtain new orbitals until self-consistency. 
    Each KS iteration contributes to the loss $\mathcal{L}$, resulting in regularization.~\cite{ksr}
    Trainable parameters $\theta$ are updated by a classical optimizer. 
    (3) The converged parameters $\theta^{*}$ define the trained functional $f_{\theta^{*}}$, which is subsequently used to evaluate the properties of any molecule in stage (4).
    Brown arrows in (2) indicate the gradient flow through the procedure. 
    CPU and QPU denote classical (i.e., central and graphical) and quantum processing units.
    }
    \label{fig:master}
\end{figure}

\begin{figure*}
    \centering
    \includegraphics[width=2.0\columnwidth]{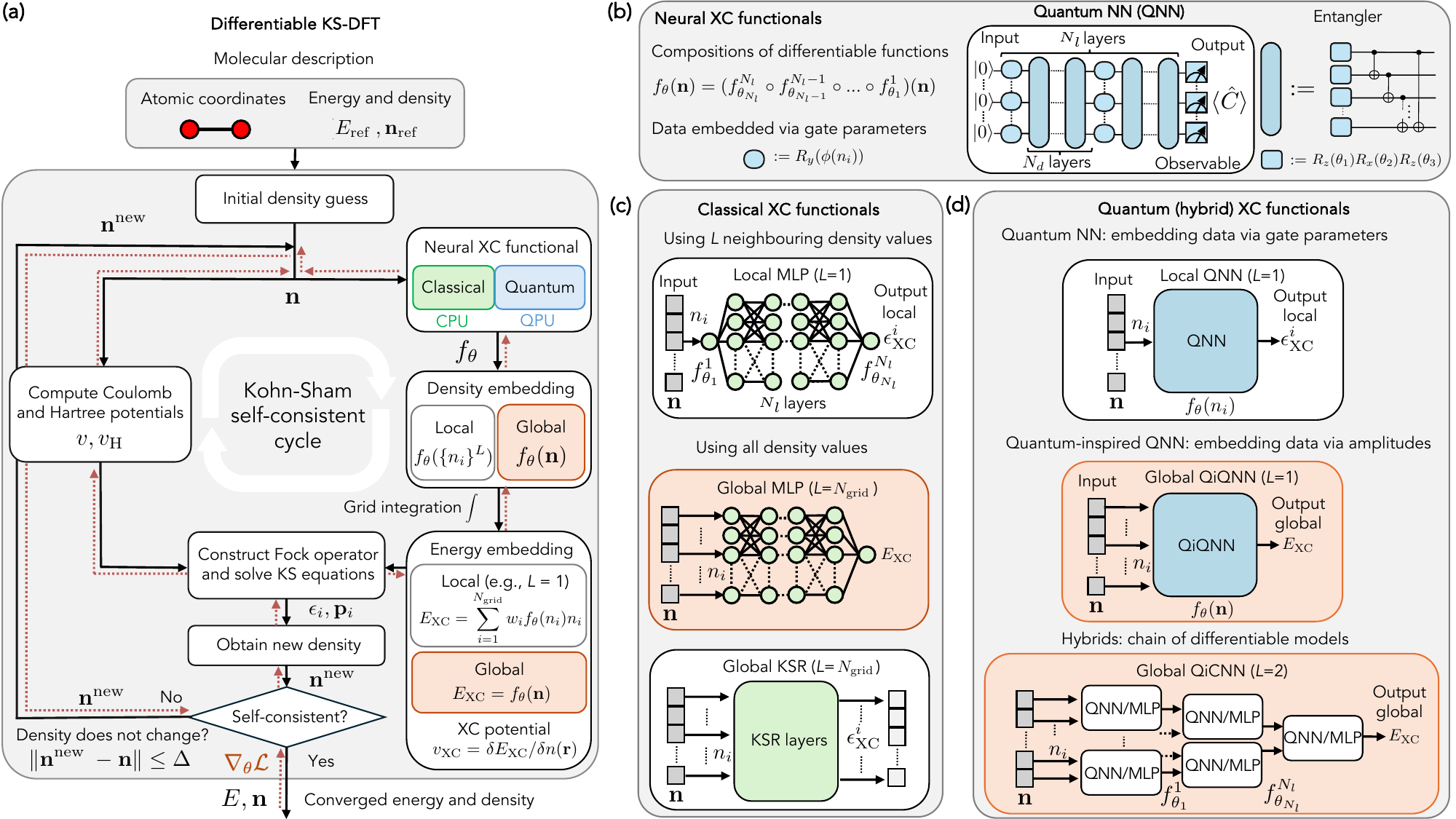}
    \caption{
    (a) Method for training quantum-enhanced neural XC functionals within KS-DFT. Data describing a single molecule is passed to KS-DFT. An initial electron density $\mathbf{n}$ estimate is made and potentials $v, v_{\mathrm{H}}$ and $ v_{\mathrm{XC}}$ are computed to construct the Fock operator. Neural XC functional $f_{\theta}$ is represented by a quantum and/or classical model, which takes $L$ density values as inputs (i.e., $L=1$ fully local and $L=N_{\mathrm{grid}}$ global). The XC functional can be embedded in DFT locally (outputs the values on every grid point) or globally (outputs the total XC energy). Then, the KS equation and the orbital coefficients $\mathbf{p}$ are updated. The new density is computed and if the change is small enough, the procedure is stopped (self-consistency) and the trained XC functional is obtained. Otherwise, the new orbitals are kept for the next iteration. The loss function is computed and differentiated with respect to the model parameters. The red dashed arrows show the gradient flow. (b) Building blocks of neural XC functionals are defined as compositions of differentiable functions, which can be represented by QNNs. (c) Architectures of classical and (d) quantum (hybrid) XC functionals. Models differ by how many inputs can be taken at a time. The QiCNN architecture allows one to chain classical and quantum models. The orange color indicates that the global energy embedding (see (a)) is used, and the blue color indicates quantum operations. Note that MLPs can be replaced by any classical architecture.
    }
    \label{fig:framework-detailed}
\end{figure*}

\subsection{Training neural XC functionals}

In this work, we adopt the standard supervised ML approach. 
We generate the training (reference) dataset $D = \{ (E^{i}_{\mathrm{ref}}, n^{i}_{\mathrm{ref}}) \}_{i=1}^{N_{\mathrm{mol}}}$ using an accurate \textit{ab initio} method for $N_{\mathrm{mol}}$ different molecules (also including different geometries of the same molecules).
For the 1D implementation, we use the original training and test datasets by Li~\etal~\cite{ksr} ~composed of hydrogen atom-based systems, e.g., H$_2$, H$_2$H$_2$ and H$_4$. 
The data is obtained with DMRG simulations, where $N_{\mathrm{grid}}=513$ is fixed for all the systems.
For the 3D implementation, we use the PySCF software package and the coupled cluster singles and doubles (CCSD)~\cite{bartlett2007coupled} method to generate our data for H$_2$, which is equivalent to FCI (an exact solution within a basis set) in this case.
When more than two electrons are present, we suggest using the singles and doubles coupled cluster method with perturbative triples (CCSD[T]), which is viewed as the ``gold standard" in computational chemistry for weakly correlated systems.~\cite{bartlett2007coupled}
In the presence of strong correlation, the aforementioned CC methods often fail to converge to the correct solution.~\cite{bulik2015can} 
Consequently, we employ FCI in challenging cases, such as the square H$_4$ system, for which it remains feasible given the small system sizes considered in this work.
A Gaussian-type orbital basis set, 6-31G, is used in all calculations.
For an integration grid, we use SG-3 with a Becke scheme to combine multi-atomic grids, such that for a given molecule at all geometries we obtain the same number of grid points.~\cite{dasgupta2017standard}
The grid density is set to minimum to reduce the number of neural XC functional inputs, which simplifies the complexity of the problem.
This entails $N_{\mathrm{grid}}=1192$ for H$_2$ and $N_{\mathrm{grid}} = 2384$ for planar H$_4$ in the 3D setting.
Given that we train models that allow only for a fixed number of inputs ($N_{\mathrm{grid}}$), grids have to be adapted (e.g., by pruning) to have the same number of grid points for all molecules in the dataset.~\cite{dasgupta2017standard}
However, this is only a limitation imposed by simple prototypical (quantum) models that are investigated in this work.
See Appendix~\ref{app:grids} for additional clarifications about integration grids.

To train our models, we use the loss function given by
\begin{equation}
\label{eq:loss}
\begin{aligned}
\mathcal{L}(\theta)= & \underbrace{\frac{1}{N_{\mathrm{mol}}}\sum^{N_{\mathrm{mol}}}_{i=1}\left[\sum_{j=1}^{N_{\mathrm{grid}}}w_j\left(n^{i}_{\mathrm{KS, j}}(\theta)-n^{i}_{\mathrm{ref}, j}\right)^2 / N^i_e\right]}_{\text {density loss } \mathcal{L}_n} \\
& +\underbrace{\frac{1}{N_{\mathrm{mol}}}\sum^{N_{\mathrm{mol}}}_{i=1}\left[\sum_{k=1}^{N_{\mathrm{KS}}} \eta_k\left(E^{i}_k(\theta)-E^{i}_{\mathrm{ref}}\right)^2 / N^{i}_e\right]}_{\text {energy loss } \mathcal{L}_E},
\end{aligned}
\end{equation}
where the expectation is taken on $N_{\mathrm{mol}}$ molecules of the training set, $\eta_k$ weights the contribution of every energy at each KS iteration and $w_j$ represents the integration grid weight.
We use the weighting scheme introduced by Li~\etal~\cite{ksr}:
$\eta_k=0.9^{N_{\mathrm{KS}}-k} H(k-10)$, where $H$ is the Heaviside function to tune the contribution of every loss term.
This composite loss function includes two primary components: the density loss $\mathcal{L}_n$ and the energy loss $\mathcal{L}_E$. 
These components are essential for guiding the neural network to accurately produce electron densities and total energies, respectively. 
Specifically, the density loss assesses the accuracy of the Kohn-Sham electron density produced by the DFT pipeline compared to the reference electron density. 
Meanwhile, energy loss measures the precision of the Kohn-Sham energies generated by the neural network against the reference energies at each iteration $k$. 
The energy loss is weighted by $\eta_k$ to emphasize the importance of later iterations.

All derivatives of the loss function are computed via JAX automatic differentiation in this work.
On quantum hardware, exact differentiation schemes such as parameter-shift rule~\cite{schuld2019evaluating} or its generalized version~\cite{kyriienko2021generalized} can be used (see Appendix~\ref{app:diff}).
The loss function is then minimized to obtain the optimal model parameters
$$
\theta_{o p t}=\operatorname{argmin}_\theta \mathcal{L}(\theta)
$$
that define the optimal XC correlation functional. 
In the following scenario, we outline the use of a gradient-based method for training, although it should be noted that our approach is also compatible with gradient-free techniques, which would, however, prevent regularization as in ref.~\onlinecite{ksr}.

At each converged SCF calculation, the gradient of loss $\mathcal{L}$ is performed via the automatic differentiation (or potentially parameter-shift rule on a quantum device). 
The gradient is then provided to the classical optimizer (Adam~\cite{kingma2014adam} or L-BFGS-B~\cite{zhu1997algorithm}). 
The learning rate for the Adam optimizer can be decreased from $10^{-2}$ to $10^{-4}$ as the training progresses to fine-tune the results. 
Otherwise, L-BFGS-B and Adam are used with the standard hyperparameters.

\section{\label{sec:results} Results and discussion}

\begin{figure*}[ht]
    \centering
    \includegraphics[width=1.95\columnwidth]{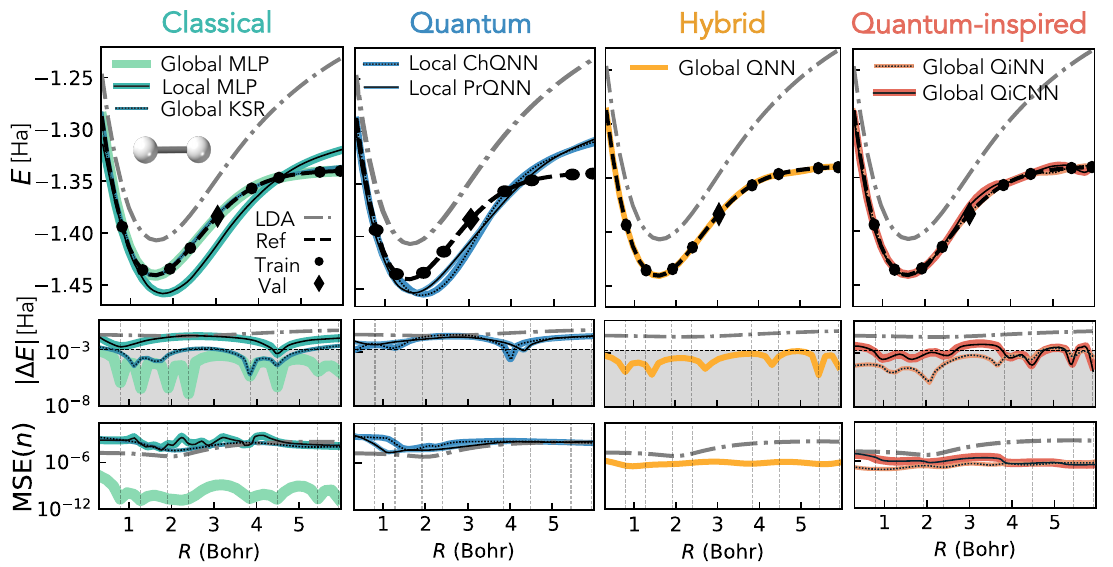}
    \caption{
    Results of training neural XC functionals on different geometries of hydrogen molecule in 1D KS-DFT framework~\cite{ksr}. 
    The grey area highlights results within chemical precision ($\leq 1.6$ mHa) and vertical dashed lines mark the training data.
    Local models improve upon the LDA shape, but still produce a poor qualitative agreement with the reference.
    Best results are obtained for global models when the whole density is considered as model input.
    The XC models are defined in Section~\ref{sec:models}.
    All neural functionals are trained on eight training data points (circles) and the final model parameters are selected according to the lowest loss on the validation results.
    The reference (Ref) results are generated using DMRG by Li \textit{et al.}~\cite{ksr}
    }
    \label{fig:h2-1D}
\end{figure*}

In this section, we present the outcomes of numerical experiments and comment on their performance in 1D and 3D settings.
First, we evaluate all our models on H$_2$ using the 1D KS-DFT framework based on the implementation of Li~\etal~\cite{ksr} 
Next, we extend our approach to the 3D framework based on PySCFAD~\cite{zhang2022differentiable} and test the extrapolation capacity of our better models on the H$_2$ molecule.
We assess the effect of the fractional orbital occupation on H$_2$ and evaluate our best model on planar H$_4$.
We evaluate the generalization capability on the H$_2$H$_2$ system, which was not part of the training dataset, in a 1D setting.
Finally, we estimate the number of measurements required in a quantum hardware implementation.
All calculations are performed on a classical computer using Qadence.~\cite{qadence2024pasqal}

For all our results, we define the density error in terms of the mean squared error (MSE) and include the local-density approximation (LDA) functional to contrast with our ML-based functionals.
In addition, we analyze the results in terms of the non-parallelity error (NPE)~\cite{tuna2015assessment}, defined as the absolute difference between the largest and smallest energy errors $\Delta E$ over the whole energy profile. 
The average $|\Delta E|$ measures the overall closeness to the true dissociation profile, whereas the NPE measures non-parallelism after removing any constant shift. 
In reaction rate calculations, the absolute energy scale is irrelevant; only the relative shape of the energy profile matters, e.g., see ref.~\cite{douroudgari2021atmospheric}.
For example, one could have $\text{average}~|\Delta E| = 1$ Ha (constant shift) but $\text{NPE} < 1$ mHa, indicating that the correct chemical behavior (e.g., reaction rates) is preserved, enabling accurate predictions. 
Regarding the MSE on the density, in ref.~\cite{kim2013understanding} it was shown that, in typical cases, density errors do not significantly affect the total energy, except in specific scenarios. 
Therefore, in practical applications, NPE is often the most important metric.

The hyperparameter search is performed for all proposed models and the best parameter combination (e.g., number of qubits $N_q$, layers $N_l$ and depth $N_{d}$) is selected.

\begin{table*}[ht]
\centering
\caption[]{Extrapolation results of neural XC functionals evaluated on H$_2$ within 1D KS-DFT framework (Fig.~\ref{fig:h2-1D}).
Average energy and density errors, as well as non-parallelity error (NPE), are evaluated along the whole dissociation profile.
The uncertainties were calculated using the standard deviation.
For MLP-based models, we indicate the number of neurons per layer and the number of layers (e.g., $513/2$).
Description of the abbreviations: No.: number, Av.: average.
} 
\label{tab:results-1D-h2}
\tiny
\setlength{\tabcolsep}{1.3pt}
\begin{tabular*}{\textwidth}{@{\extracolsep{\fill}} ll llll llll}
\hline\hline
& & Global MLP & Local MLP & Global KSR & Local ChQNN & Local PrQNN & Global QNN & Global QiNN & Global QiCNN \\
\hline
& $N_q$/$N_d$/$N_l$ & 513/2 & 513/1 & --- & 4/8/1 & 4/8/1 & 6/8/1 & 9/10/1 & 6/8/1\\
& No. parameters & $2.64 \cdot 10^5$ & 1540 & 1568 & 96 & 108 & 316 & 297 & 596 \\
H$_2$ & Av. $|\Delta E|$ [Ha]  & $(2 \pm 2)\cdot 10^{-4}$  & $(1.5 \pm 0.9) \cdot 10^{-2}$ & $(1 \pm 1) \cdot 10^{-3}$ & $(1.5 \pm 1) \cdot 10^{-2}$ & $(1.4 \pm 0.8) \cdot 10^{-2}$ & $(4 \pm 4) \cdot 10^{-4}$ & $(8 \pm 8) \cdot 10^{-4}$ & $(3 \pm 2) \cdot 10^{-4}$\\
& NPE [Ha] & $1.3 \cdot 10^{-3}$ & $2.9 \cdot 10^{-2}$ & $4 \cdot 10^{-3}$ & $3.1 \cdot 10^{-2}$ & $3.1 \cdot 10^{-2}$ & $1.4 \cdot 10^{-3}$ & $4 \cdot 10^{-3}$ & $7 \cdot 10^{-4}$ \\
& Av. MSE($n$) & $(1.0 \pm 3.0) \cdot 10^{-9}$ & $(5.6 \pm 6.0) \cdot 10^{-4}$ & $(2.4 \pm 2.4) \cdot 10^{-4}$ & $(6.0 \pm 5.3) \cdot 10^{-4}$ & $(4.0 \pm 6.0) \cdot 10^{-4}$ & $(7.2 \pm 2.4) \cdot 10^{-7}$ & $(2.5 \pm 1.8) \cdot 10^{-3}$ & $(2.1 \pm 1.0) \cdot 10^{-7}$ \\
\hline \hline
\end{tabular*}
\end{table*}

\subsection{Local versus global functionals}

In Figure~\ref{fig:h2-1D}, we evaluate all our models within the 1D framework.
This figure compares the prediction of 1D potential energy surfaces achieved by different models, characterizing the H-H bond dissociation in the H$_2$ molecule.
We start with MLP, which is the prototypical expressive model that we use to evaluate the approaches of local and global embeddings. 
To that end, in all our MLP models, we used 2 hidden layers with 513 neurons, which matches the number of grid points, to achieve enough expressivity. 
The Local MLP model has a single input and output neuron to infer the local XC energy based on corresponding density.
Local MLP improves on the LDA DFT results in terms of overall energy error, but it is not sufficient to improve on the density.
The restriction local information is limiting the capacity of the model to improve.
Hence, we observe the best trade-off of energy-density error that this model is capable of.

In order to evaluate if local quantum models can improve on the Local MLP, we propose the Local ChQNN and PrQNN. 
These models employ the Chebyshev and product feature maps (see Appendix~\ref{app:qnn}), respectively.
We use $N_q = 4$ qubits, $N_l=1$ (no data reuploading) and $N_d=8$ layers for our ansatz.
The results in terms of NPE of the energy are very similar to those of the Local MLP model. 
The energy and density profiles of local quantum and classical models are qualitatively similar, as expected due to locality of both models. 
It is therefore necessary to verify if including more information in the quantum models brings benefits in terms of lower energy error $|\Delta E|$ and MSE of the density profile. 

For that purpose, we evaluate a global version of MLP with $N_{\mathrm{grid}}=513$ input neurons. 
This allows the model to consider the whole density to infer the XC energy via the global embedding, hence there is only one output neuron.
This results in $\sim 2.6 \cdot 10^{5}$ parameters as reported in Table~\ref{tab:results-1D-h2}. 
The corresponding energy profile is significantly improved (lower error with respect to the reference) in comparison to all other models.
Similarly, the proposed KSR model~\cite{ksr} can be tailored to any locality.
The global version of KSR model was shown to perform best and generalize perfectly in a single molecule even from only 2 data points (here 8).
We used the same density initialization as for Global MLP and the model struggled to generalize correctly, leading to MSE on density comparable to LDA, despite rendering chemically accurate energies.
However, KSR is significantly more parameter-efficient, having only $\sim 10^3$ parameters versus $\sim 2.6 \cdot 10^{5}$ for Global MLP, which are generally known to overfit the training data.
In addition, given good initialization, global KSR is capable of learning the dissociation profile of H$_2$ perfectly only from 2 data points, as shown in ref.~\onlinecite{ksr}.
For all the remaining models with global embedding, the results perform as well or even better than local LDA-like functionals, at least for the total energies. 
Moreover, the global models have the lowest non-parallelity error (NPE) by an order of magnitude in comparison to local models (see Tab.~\ref{tab:results-1D-h2}). 
A significant reduction (3-5 orders of magnitude) in average MSE of the electron density is achieved only by fully global models such as Global QNN and Global MLP.

Our results indicate that both classical and quantum models must incorporate global electron density features to achieve chemical precision.

\subsection{Extrapolation within a single molecule}

All our models provide smooth energy profiles due to regularization capabilities of the employed differentiable DFT implementation.
In Fig.~\ref{fig:h2-1D}, it is clear that 8 training data points, sampling low to strongly correlated regimes of the hydrogen molecule, are sufficient to generalize correctly within a single molecule, provided the model has enough expressivity.
In general, all the models show characteristic dips in energy errors at the training data points.
However, the density MSE profiles are rather flat along the whole dissociation curve.
Similarly, in both Figs.~\ref{fig:h2-3D} and~\ref{fig:h4}, just 3 training points were sufficient to train models that provide qualitatively correct energy profiles, even in the presence of strong correlation.
Further studies are necessary to investigate the extrapolation capabilities of quantum models on large datasets, i.e., examining whether models trained only in strongly correlated regimes can extrapolate correctly to low-correlation regimes, and vice versa.

\subsection{Hybrid and quantum-inspired quantum-classical models}

Since local QNNs are not capable of correctly fitting even the training data, we explore architectures that allow us to embed $L$ data points and hence approach non-local functionals.
In particular, in the Global QNN model, we consider the idea of embedding $L=3$ data points in a single QNN with $N_q = 6$ qubits, $N_l=1$ (no data reuploading) and $N_d=8$ layers, and apply the same model onto $N_{\mathrm{grid}} / L $ different batches. 
The results are combined via a single-layer MLP with 171 input neurons, one bias parameter, and one output for the global embedding.
This allows us to embed chunks of data into a limited quantum model with a restricted number of qubits.
By combining such outputs, we obtain a fully global functional.
The model has only $\sim 300$ trainable parameters, which is orders of magnitude smaller than for the Global MLP,  since both are composed of most general quantum/classical layers.
Note that comparing it with the KSR model is not reasonable because the latter possesses numerous inductive biases such as the self-interaction gate and the softly imposed negativity of XC energy.

We also evaluated the amplitude embedding as a feature map, using a DQC with $N_q = 9$ qubits, $N_l=1$, $N_d=10$.  
We name this model Global QiNN since, in general, it is hard to embed the classical data (electron density) exactly in a polynomial number of one-/two-qubit gates.~\cite{park2019circuit}
The number of qubits was chosen to fit all the elements of $\mathbf{n}$ exactly in a state-vector.
In Tab.~\ref{tab:results-1D-h2}, the results for Global QiNN are in good agreement with the reference: within chemical precision and overall density error that significantly improved upon LDA. 
Notably, the outcome is similar to Global QNN, which is simple to implement on near-term devices.
The Global QiNN architecture, despite producing satisfactory results, does not produce a significant advantage over quantum-classical hybrid QNN.
It can, however, serve for theoretical studies since only quantum layers are employed, which dissociates it completely from classical models.

Another approach is to chain quantum layers with angle embeddings, as implemented in the Global QiCNN model.
In that case, the DQC output becomes the input of the next DQC layer.
We use three different DQCs that all employ $N_q = 6$ qubits, $N_l=1$ and $N_d=8$ ansatz layers.
The third DQC layer feeds into an MLP layer with 19 neurons (remaining outputs of DQC at the third layer) to combine the results into the total XC energy.
The results for the energy are within chemical precision. 
Concerning the density, the errors are significantly lower, at least by an order of magnitude, than LDA.

Overall, among the hybrid quantum models presented, the simplest, Global QNN with a single quantum and classical layer, yields the best performance.

\begin{figure}[ht]
    \centering
    \includegraphics[width=0.95\columnwidth]{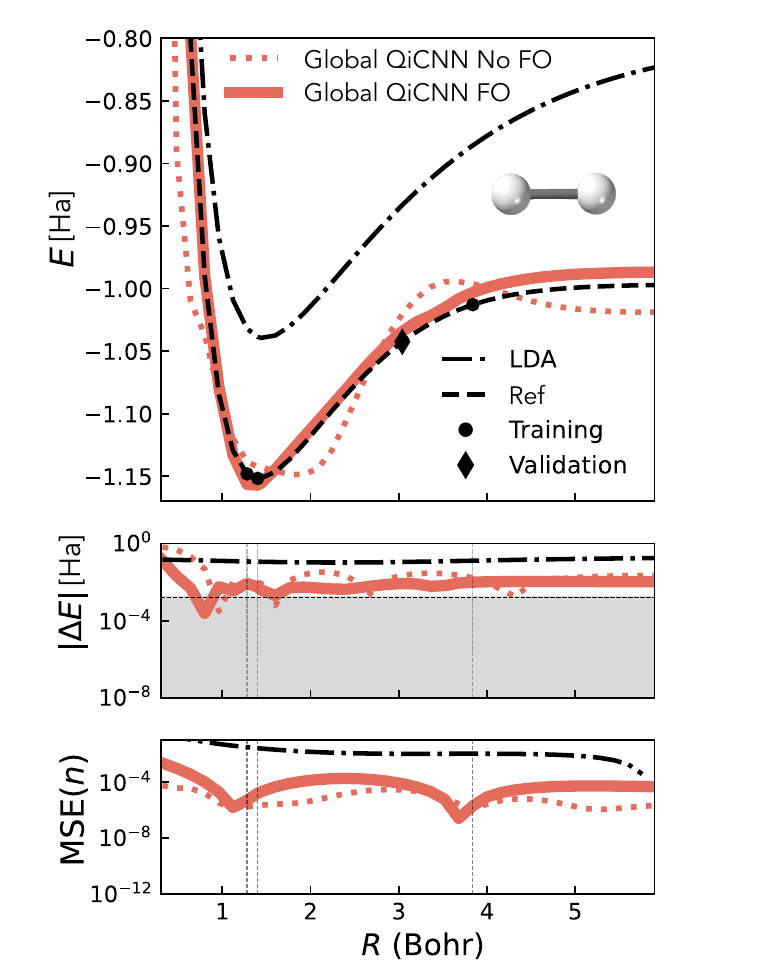}
    
    \caption{
     Results of training neural XC functionals on different geometries of H$_2$ in our 3D RKS-DFT framework based on PySCFAD. 
     The same Global QiCNN (Eq.~\eqref{eq:qicqnn}) is trained with and without fractional orbital occupation (FO).
    Despite the limited training data (3 points), the model accurately extrapolates to distant geometries, demonstrating the strong regularization capacity of our 3D framework.
    The reference (Ref) results are generated using the FCI/6-31G level of theory.
    }
    \label{fig:h2-3D}
\end{figure}

\begin{table}[ht]
\centering
\caption[]{Extrapolation results for Global QNN evaluated on H$_2$ within 3D KS-DFT framework, both with and without fractional orbital occupation (FO) (Fig.~\ref{fig:h2-3D}).
} 
\label{tab:results-3D-h2}
\smaller
\setlength{\tabcolsep}{2.5pt}
\begin{tabular*}{\columnwidth}{llll}
\hline\hline
& & Global QiCNN No FO & Global QiCNN FO \\
\hline
& $N_q$/$N_d$/$N_l$  & 4/14/1  & 4/14/1  \\
& No. parameters    & 653   & 653 \\
 H$_2$ & Av. $|\Delta E|$ [Ha]  &  $(5 \pm 13) \cdot 10^{-2}$  & $(1 \pm 2) \cdot 10^{-2}$ \\
& NPE [Ha] & $6.5 \cdot 10^{-1}$  & $1.3 \cdot 10^{-1}$ \\
& Av. MSE($n$) & $(1.1 \pm 1.3) \cdot 10^{-5}$  & $(1.5 \pm 3.7) \cdot 10^{-4}$ \\
\hline\hline
\end{tabular*}
\end{table}

\subsection{Extension from 1D to 3D KS-DFT framework}

To demonstrate the ability to efficiently train quantum XC models in our 3D framework, we evaluate our most complex neural functional, Global QiCNN, in such a setting.
Similarly to the 1D version, the model consists of 3 chained quantum layers, each with $N_q = 4$ qubits, $N_l=1$ and $N_d=14$ ansatz layers.
The final quantum layer connects to MLP with 149 neurons to output the total XC energy.
As depicted in Fig.~\ref{fig:h2-3D}, Global QiCNN suffers from energy oscillations and is unable to extrapolate correctly.  
This is the ideal situation to demonstrate the ability of fractional occupation (FO) methods to alleviate such problems, which are apparent in 3D due to increased complexity.
To that end, we show that even with 3 data points (of which 2 are almost at the same bond distance) the complex model can be regularized in order to fit better the dissociation curve, e.g., lower NPE (see Tab.~\ref{tab:results-3D-h2}).
Despite not being within chemical precision along the whole energy profile, the density error is improved by orders of magnitude in comparison to the LDA result.
This is achieved with a FO algorithm~\cite{frac-occ} that is discussed in Appendix~\ref{app:frac-occ}.
The FO approach relaxes the constraint of the integer orbital occupation in PySCF's restricted KS-DFT implementation and allows us to approximate the true fractional ones.
In addition, the neural XC functional is optimized through training, thereby compensating for errors arising from the approximation of orbital occupations. 
As reported in Tab.~\ref{tab:results-3D-h2}, the use of the FO approach results in an approximately 5-fold reduction in both the average energy error and its variance, without any extra model parameters.
Note that the FO method is not limited to use only with quantum models; it can also be applied more broadly for training classical neural XC functionals.

As the ultimate test of our 3D framework, demonstrating the capability to accurately learn from challenging molecular systems, we trained our best hybrid architecture, the Global QNN, on planar H$_4$. 
This is a typical small system exhibiting strong correlation, commonly used to benchmark \textit{ab initio} methods.~\cite{bulik2015can, sokolov2020quantum}
Here, Global QNN consists of one quantum layer with $N_q = 4$ qubits, $N_l=1$, $N_d=14$ and followed by MLP with 596 neurons that outputs $E_{\mathrm{XC}}$.
In Fig.~\ref{fig:h4}, we note that the LDA approximation greatly overestimates the energy barrier at $\beta = 90^{\mathrm{o}}$.
This quantum-classical hybrid model, given its limited expressivity, provides the best achievable fit of the training data and extrapolates smoothly to various geometries.
In Tab.~\ref{tab:results-3D-h4}, we note that the model achieves chemical precision over the whole energy profile (NPE $\approx 10^{-4}$). 
Concerning density errors, they are of the order of $10^{-5}$, similar to the H$_2$ density errors (Tab.~\ref{tab:results-3D-h2}) that are obtained within the same 3D framework.
Hence, we demonstrated that provided accurate reference data, it is possible to train quantum-enhanced functionals to learn from strongly correlated systems. 
In the future, it would be interesting to evaluate if such models can generalize correctly from small to other larger strongly correlated systems such as stretched N$_2$.

\begin{figure}[ht]
    \centering
    \includegraphics[width=0.9\columnwidth]{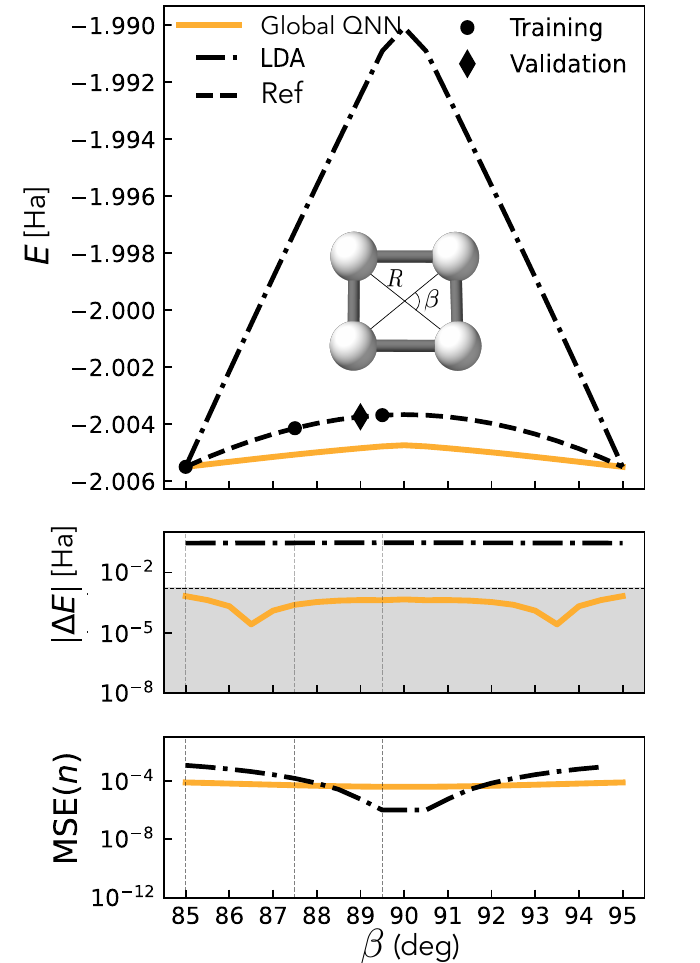}
    \caption{
    Extrapolation of Global QNN (Eq.~\eqref{eq:globalqnn}) on planar H$_4$ within the 3D framework with a fractional occupation approach.~\cite{frac-occ} 
    Global QNN achieves chemical precision in energy, surpassing the LDA results by orders of magnitude.   
    In the top panel, LDA and Global QNN results are shifted to match the reference data at extremities to compare the shapes of the dissociation profiles. 
    Shifts can be inferred from the $|\Delta E|$ panel.
    The reference (Ref) results are generated with the FCI/6-31G level of theory.
    }
    \label{fig:h4}
\end{figure}

\begin{table}[ht]
\centering
\caption[]{Extrapolation results for Global QNN evaluated on H$_4$ within the 3D KS-DFT framework (Fig.~\ref{fig:h4}).
} 
\label{tab:results-3D-h4}
\smaller
\setlength{\tabcolsep}{7.5pt}
\begin{tabular*}{0.8\columnwidth}{lll}
\hline\hline
& & Global QNN \\
\hline
& $N_q$/$N_d$/$N_l$  & 4/14/1  \\
& No. parameters    & 764 \\
 H$_4$ & Av. $|\Delta E|$ [Ha]  &  $(3 \pm 2) \cdot 10^{-4}$ \\
& NPE [Ha] & $6 \cdot 10^{-4}$ \\
& Av. MSE($n$) & $(5.3 \pm 1.2) \cdot 10^{-5}$ \\
\hline

\hline\hline
\end{tabular*}
\end{table}

\subsection{Zero-shot generalization of XC functionals}

In Fig.~\ref{fig:h2-1D}, we demonstrated that quantum-enhanced XC models can learn to extrapolate correctly to different geometries of a single molecule by consistently achieving chemical precision within the 1D framework.
To evaluate their capacity to generalize to different molecular systems, we select our best-performing quantum-enhanced functional and test it on a molecular system that was not included in the training dataset.

\begin{figure}[ht]
    \centering
    \includegraphics[width=0.85\columnwidth]{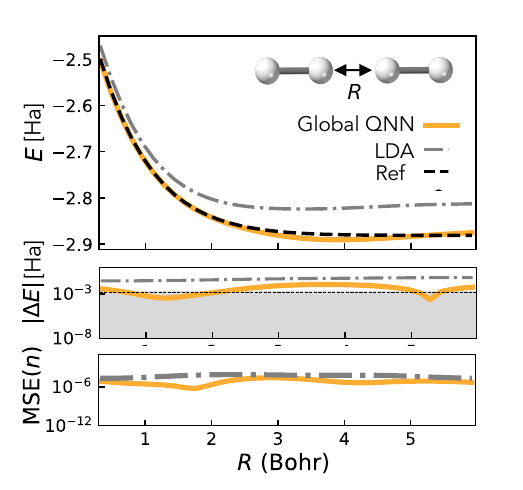}
    \caption{
    Generalization of Global QNN (Eq.~\eqref{eq:globalqnn}) to unseen H$_2$H$_2$ system trained on H$_2$ and H$_4$ molecules within the 1D framework. 
    The model achieves chemical precision over some regions of the dissociation profile without being explicitly trained on it.
    The reference (Ref) results are generated using DMRG by Li \etal~\cite{ksr}
    }
    \label{fig:h2h2-infer}
\end{figure}

Our results highlight Global QNN as the simplest hybrid model, comprising a single quantum and classical layer, capable of achieving chemical precision for H$_2$.
For its training, the dataset only contained the hydrogen molecule at different geometries. 
Such a dataset does not allow the model to implicitly learn the XC energy dependency on the variable number of electrons $N_e$ in the input density (i.e., $N_e = \sum_{i=1}^{N_{\mathrm{grid}}}w_i {n}_i$) when evaluating on a different system such as H$_2$H$_2$ (separate hydrogen molecules).
Hence, we enrich our train dataset with molecules such as linear H$_4$ at different bond distances, as done in ref.~\cite{ksr}
We train Global QNN using the following bond distances: For H$_2$, $R \in (0.48, 1.28, 3.04, 3.84)$ Bohr and for H$_4$, $R \in (1.28, 2.08, 3.36, 4.48)$ Bohr.

\begin{table}
\centering
\caption[]{Zero-shot generalization results for Global QNN evaluated on H$_2$H$_2$, trained on H$_2$ and H$_4$ within the 1D KS-DFT framework (Fig.~\ref{fig:h2h2-infer}).
}
\label{tab:results-1D-gen}
\smaller
\setlength{\tabcolsep}{7.5pt}
\begin{tabular*}{0.8\columnwidth}{lll}\hline\hline
& & Global QNN \\
\hline
& $N_q$/$N_d$/$N_l$ & 6/8/1 \\
 & No. parameters  & 316  \\
 H$_2$H$_2$ & Av. $|\Delta E|$ [Ha] & $(4.9 \pm 1.4) \cdot 10^{-2}$  \\
& NPE [Ha] &  $3.9 \cdot 10^{-2}$ \\
& Av. MSE($n$)  &  $(5 \pm 2) \cdot 10^{-5}$ \\
\hline\hline
\end{tabular*}
\end{table}

In Fig.~\ref{fig:h2h2-infer}, we note that chemical precision is achieved only for small values of $R < 2$ Bohr and momentarily at $R = 5.3$ Bohr, with the overall energy closely aligned with the reference, resulting in lower energy error compared to LDA. 
MSE on the density consistently remains lower than the LDA result.
The XC model yields suboptimal qualitative predictions, evidenced by the energy minimum at approximately $3.5$ Bohr (stable configuration). 
Considering the utilization of only basic classical and quantum layers and a few data points, these results demonstrate the capacity of our models to generalize. 
In fact, as stated in Tab.~\ref{tab:results-1D-gen}, our model has only 316 trainable parameters compared to the state-of-the-art Global KSR model, which has $\approx 6$ times more.
Hence, addressing the shortcomings would require improving the architecture of our model (e.g., using molecular symmetries) and expanding the training dataset, which is beyond the scope of this study.

To summarize, we visualize all our results for the energy error of quantum(-enhanced) models in Fig.~\ref{fig:results-summary}.
Overall, the use of quantum layers can reduce the number of trainable parameters necessary to achieve a given precision in energy.

\begin{figure}[ht]
    \centering
    \includegraphics[width=0.8\columnwidth]{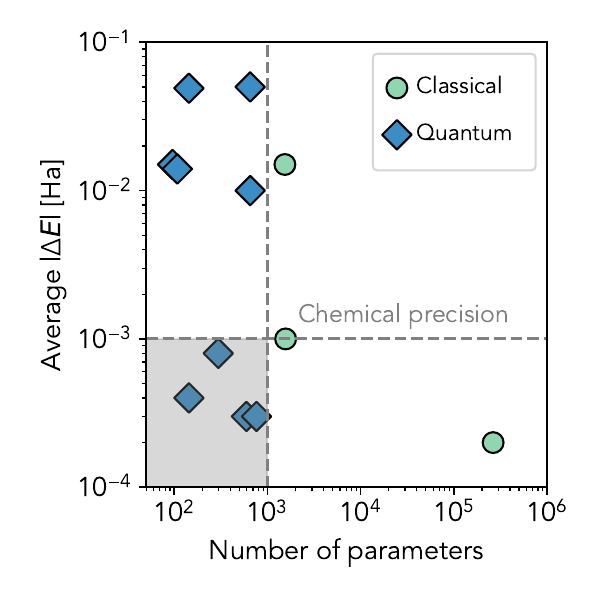}
    \vspace{-0.5cm}
    \caption{
    Dependence of the average energy error on the size of classical and quantum (hybrid) models. 
    Quantum models are consistently more compact ($\leq 10^{3}$ parameters) than classical models (e.g., GKSR with $1568$ parameters) given a fixed energy error, while being still capable of achieving chemical precision.
    }
    \label{fig:results-summary}
\end{figure}

\subsection{Performance evaluation of QNNs and MLPs in an identical backbone architecture}\label{sec:qnn-vs-mlp}

In the previous sections, we examined several QNN architectures. Among them, the hybrid Global QNN (GQNN) demonstrated the best performance for the single-molecule case (Fig.~\ref{fig:h2-1D}) and exhibited promising zero-shot generalization capabilities on the H$_2$H$_2$ system (Fig.~\ref{fig:h2h2-infer}). 
However, the extent to which QNNs contribute performance gains over analogous classical architectures within the same backbone remains unclear. 

To address this, we substitute the quantum layers in the GQNN architecture with standard MLPs, naming that model Convolutional MLP (CMLP), shown in Fig.~\ref{fig:mlp-vs-qnn}. 
For both classical and quantum models, we preserve the overall model structure and maintain a similar parameter count to make the best possible comparison.
In particular, the GQNN (CMLP) model is composed of a QNN (MLP) module that accepts $N_q$ ($N_i$)  inputs, applied on $N_{\mathrm{grid}} / N_{q}$ ($N_{\mathrm{grid}} / N_{i}$) times on the ${N_{\mathrm{grid}}}$ input density. 
Those outputs are combined via a linear layer with a bias, producing the total XC energy.
The quantum model has $N_q$ qubits in both the input layer (Product feature map, see Sec.~\ref{sec:feature_maps}) and each of the $N_l$ hidden layers.
MLP uses $N_i$ neurons in the input layer with $N_q$ neurons in the hidden layers.

\begin{figure}[ht]
    \centering
    \includegraphics[width=0.8\columnwidth]{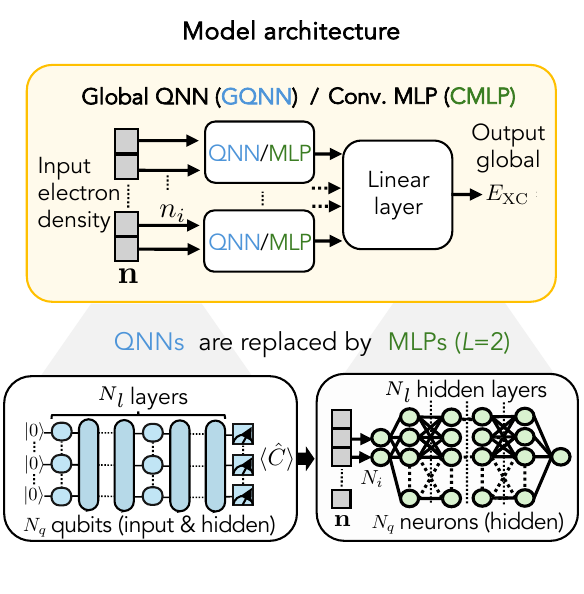}
    \vspace{-0.5cm}
    \caption{
    Architecture of the Global QNN (GQNN) and Convolutional MLP (CMLP) models. 
    To assess the impact of quantum layers relative to classical ones, we replace the QNNs in the GQNN backbone with MLPs (CMLP).
    }
    \label{fig:mlp-vs-qnn}
\end{figure}

We compare GQNN and CMLP for the interpolation and the generalization tasks.
The results are assembled in Tab.~\ref{tab:mlp_vs_qnn} for the interpolation of energies across various bond distances of the hydrogen molecule, and in Tab.~\ref{tab:mlp_vs_qnn_gen} for generalization to a different molecular system, H$_2$H$_2$, after training on H$_2$ and H$_4$.
We use the same datasets and hyperparameters as in previous experiments shown in Figs.~\ref{fig:h2-1D} and~\ref{fig:h2h2-infer}, respectively.

The performance of CMLP-based XC functional strongly depends on the choice of activation functions. 
We employ the hyperbolic tangent (tanh) and softplus (soft) activation functions to ensure infinite differentiability -- necessary for computing molecular properties -- and to provide smooth output suitable for correct exchange–correlation (XC) potential representation~\cite{kasim2021learning}.
To control the expressivity of the XC functional, we tune the number of inputs per block -- denoted $N_q$ for GQNN and $N_i$ for CMLP to -- as well as the number of layers $N_l$ within each QNN or MLP module (see Fig.~\ref{fig:mlp-vs-qnn}).

\begin{table}
    \centering
    \caption[]{
    Interpolation: Comparison of GQNN and CMLP (Fig.~\ref{fig:mlp-vs-qnn}) in fitting the dissociation profile of the hydrogen molecule.  
    The results are obtained within the 1D framework using the same set of molecular geometries as in Fig.~\ref{fig:h2-1D}. 
    The best result for each metric is highlighted in bold. 
    Description of abbreviations: Act.: activation function; Par.: number of parameters.
    }
    \label{tab:mlp_vs_qnn}
    \tiny
    \setlength{\tabcolsep}{1.5pt}
    \begin{tabular*}{\columnwidth}{lll lll}
    \hline\hline
    {Model} & {$N_i|N_q|N_l|$Act.} & Par. & {$\Delta E$ [Ha]} & Av. MSE($n$) & NPE [Ha] \\
    \hline
GQNN & 3/6/8 & 316 & $\mathbf{(4.0 \pm 4.0)\cdot10^{-4}}$ & $(7.2 \pm 2.4)\cdot10^{-7}$ & $\mathbf{1.40\cdot10^{-3}}$ \\
CMLP & 2/2/1/tanh & 521 & $(1.7 \pm 1.1)\cdot10^{-3}$ & $(1.3 \pm 1.0)\cdot10^{-5}$ & $3.38\cdot10^{-3}$ \\
CMLP & 2/2/1/soft & 521 & $(5.9 \pm 6.0)\cdot10^{-3}$ & $(9.1 \pm 4.6)\cdot10^{-6}$ & $2.48\cdot10^{-2}$ \\
CMLP & 19/19/1/tanh & 428 & $(1.6 \pm 1.9)\cdot10^{-3}$ & $(5.4 \pm 1.0)\cdot10^{-6}$ & $5.57\cdot10^{-3}$ \\
CMLP & 19/19/1/soft & 428 & $(7.4 \pm 34.8)\cdot10^{-2}$ & $(2.0 \pm 1.4)\cdot10^{-4}$ & $2.01\cdot10^{0}$ \\
CMLP & 3/6/4/tanh & 329 & $(1.6 \pm 1.2)\cdot10^{-3}$ & $(5.5 \pm 2.4)\cdot10^{-7}$ & $3.49\cdot10^{-3}$ \\
CMLP & 3/6/4/soft & 329 & $(9.5 \pm 11.5)\cdot10^{-4}$ & $\mathbf{(1.4 \pm 0.8)\cdot10^{-7}}$ & $4.79\cdot10^{-3}$ \\
CMLP & 4/4/1/tanh & 193 & $(7.5 \pm 6.5)\cdot10^{-3}$ & $(8.2 \pm 6.4)\cdot10^{-5}$ & $2.14\cdot10^{-2}$ \\
CMLP & 4/4/1/soft & 193 & $(1.6 \pm 1.7)\cdot10^{-3}$ & $(7.8 \pm 1.8)\cdot10^{-6}$ & $7.96\cdot10^{-3}$ \\
\hline\hline
    \end{tabular*}
\end{table}

\begin{table}
    \centering
    \caption[]{
    Generalisation: Comparison of GQNN and CMLP (Fig.~\ref{fig:mlp-vs-qnn}) for zero-shot generalisation. 
    Results are evaluated within the 1D framework on H$_2$H$_2$, trained on H$_2$ and H$_4$, exactly as in Fig.~\ref{fig:h2h2-infer}. 
    The best result for each metric is highlighted in bold. 
    Description of abbreviations: Act.: activation function; Par.: number of parameters.
    }
    \label{tab:mlp_vs_qnn_gen}
    \tiny
    \setlength{\tabcolsep}{1.5pt}
    \begin{tabular*}{\columnwidth}{lll lll}
    \hline\hline
    {Model} & {$N_i|N_q|N_l|$Act.} & Par. & {$\Delta E$ [Ha]} & Av. MSE($n$) & NPE [Ha] \\
    \hline
GQNN & 3/6/8 & 316 & $(4.9 \pm 1.4)\cdot10^{-2}$ & $\mathbf{(5 \pm 2)\cdot10^{-5}}$ & $\mathbf{3.90\cdot10^{-2}}$ \\
CMLP & 2/2/1/tanh & 521 & $(3.8 \pm 2.8)\cdot10^{-2}$ & $(4.7 \pm 3.8)\cdot10^{-4}$ & $8.00\cdot10^{-2}$ \\
CMLP & 2/2/1/soft & 521 & $(4.5 \pm 4.2)\cdot10^{-2}$ & $(8.0 \pm 4.3)\cdot10^{-4}$ & $1.32\cdot10^{-1}$ \\
CMLP & 19/19/1/tanh & 428 & $(5.1 \pm 3.0)\cdot10^{-2}$ & $(1.7 \pm 1.1)\cdot10^{-4}$ & $8.69\cdot10^{-2}$ \\
CMLP & 19/19/1/soft & 428 & $(5.7 \pm 3.5)\cdot10^{-2}$ & $(6.8 \pm 5.4)\cdot10^{-4}$ & $1.05\cdot10^{-1}$ \\
CMLP & 3/6/4/tanh & 329 & $(3.3 \pm 2.5)\cdot10^{-2}$ & $(4.9 \pm 4.8)\cdot10^{-4}$ & $7.19\cdot10^{-2}$ \\
CMLP & 3/6/4/soft & 329 & $\mathbf{(2.8 \pm 1.4)\cdot10^{-2}}$ & $(1.0 \pm 0.5)\cdot10^{-4}$ & $4.62\cdot10^{-2}$ \\
CMLP & 4/4/1/tanh & 193 & $(3.5 \pm 2.7)\cdot10^{-2}$ & $(5.2 \pm 5.0)\cdot10^{-4}$ & $7.75\cdot10^{-2}$ \\
CMLP & 4/4/1/soft & 193 & $(3.5 \pm 1.9)\cdot10^{-2}$ & $(1.6 \pm 1.0)\cdot10^{-4}$ & $5.78\cdot10^{-2}$ \\
\hline\hline
    \end{tabular*}
\end{table}

In Tabs.~\ref{tab:mlp_vs_qnn} and~\ref{tab:mlp_vs_qnn_gen}, we observe that the number of trainable parameters does not directly correlate with improved model performance in terms of energy prediction $\Delta E$, average density error MSE($n$), or non-parallelity error (NPE). 
For example, the CMLP configuration (2/2/1/tanh), which has 521 parameters -- nearly twice as many as the GQNN (316 parameters) -- does not outperform the GQNN in either interpolation or generalization tasks.
One possible explanation is the limited locality $L$ of the MLP blocks in the CMLP model (see Fig.~\ref{fig:mlp-vs-qnn}), analogous to small filter sizes in convolutional neural networks. 
The overall CMLP model remains global, as it ultimately processes the entire input density through repeated application of local blocks. 
In the configuration CMLP 2/2/1 (see Tab.~\ref{tab:mlp_vs_qnn}), the model processes only $L=2$ inputs at a time, which may restrict its ability to capture nonlocal density correlations. 
Increasing the locality to $L=19$ in a shallow configuration (CMLP 19/19/1) does not yield notable improvements in MSE on density or NPE. 
Interestingly, even with a modest locality of $L=4$ (CMLP 4/4/1), the results remain of the same order of magnitude as for the $L=19$ model. 
Increasing the number of layers to 4 (CMLP 3/6/4/soft) produces the best density accuracy among the CMLP variants.
Nonetheless, as shown in Tab.~\ref{tab:mlp_vs_qnn}, the GQNN still achieves the lowest energy error $\Delta E$ and NPE overall. 
Most importantly, in the generalization setting (Tab.~\ref{tab:mlp_vs_qnn_gen}), GQNN outperforms all CMLP configurations across all metrics except for average energy error.

These results suggest that quantum-enhanced XC models can perform comparably to classical neural networks in the ideal noiseless setting, with marginal benefits on small molecular systems, such as reduced NPE. 
To assess their performance under realistic conditions, we next examine the impact of gate noise and measurement (shot) noise on our results.

\subsection{\label{sec:proofs} Error bounds on energy and density}

Ensuring the stability of SCF cycles is essential for both successful training and reliable inference with neural XC functionals. 
For QNN-based functionals on quantum hardware, measurement noise and other hardware errors perturb the energy functional and thus the SCF convergence. 
Next, we develop a theoretical framework to characterize and analyze the propagation of such errors throughout the SCF procedure.

We denote by $n(\mathbf{r})\in \mathcal{X}$ an electronic density in a suitable Banach space $\mathcal{X}$. 
We recall from Sec.~\ref{sec:theory}, the total energy is defined as
$
E[n]=T_{\mathrm{S}}[n]+V[n]+E_{\mathrm{H}}[n]+E_{\mathrm{XC}}[n]
$.
Here, $E_{\mathrm{XC}}[n]$ is the (unknown) universal exchange-correlation functional. 
Let $n^*$ be the ground-state density minimizing the total energy, 
and let $E^* = E[n^*]$ be its exact energy.
In practice, the KS equations are solved by a fixed-point iteration. 
Following the convergence analysis in refs.~\cite{liu2014convergence, liu2015analysis, woods2019computing}, we consider an SCF iteration as a map $\mathcal{T} : \mathcal{X} \to \mathcal{X}$ such that
\begin{equation}
  n^{(k+1)} \;=\; \mathcal{T}\bigl[n^{(k)}\bigr].
\end{equation}
The exact solution $n^*$ satisfies $n^* = \mathcal{T}[n^*]$.
Under typical conditions (e.g., sufficient spectral gap or use of level-shifting approach~\cite{cances2000convergence}), local convergence can be reached. 
Now suppose we replace $E_{\mathrm{XC}}$ by $\widetilde{E}_{\mathrm{XC}}$, obtaining an \emph{approximate} total energy
$\widetilde{E}[n]=T_{\mathrm{S}}[n]+V[n]+E_{\mathrm{H}}[n]+\widetilde{E}_{\mathrm{XC}}[n]$.
The iteration map becomes $\widetilde{\mathcal{T}}$. 
We denote by $\widetilde{n}$ its convergent fixed point: $\widetilde{n}=\widetilde{\mathcal{T}}[\widetilde{n}]$, and $\widetilde{E}=\widetilde{E}[\widetilde{n}]$ the final approximate ground-state energy.
To guarantee convergence, we impose the following assumptions on the components of our framework and define $||\cdot|| \equiv ||\cdot||_2$.

\begin{assumption}[Local Contractivity]\label{assump:contract}
There is a neighborhood $U$ of $n^*$ in $\mathcal{X}$, and a constant $0<\kappa<1$, such that $\| \mathcal{T}[n_1] - \mathcal{T}[n_2]\| \le \kappa\,\|n_1 - n_2\|$ for all $n_1,n_2 \in U$. Moreover, $\mathcal{T}[U]\subset U$. 
\end{assumption}

This ensures that $n^*$ is the unique fixed point in $U$ and that $\{n^{(k)}\}$ remains in $U$ under iteration. 
(See also~\cite{liu2014convergence,bai2020optimal} for the role of contractivity in SCF.)

\begin{assumption}[Uniform XC Error]\label{assump:uniform}
Fix $\varepsilon>0$ and $K>0$ such that, for all $n \in U$,
\begin{align}
\bigl|\widetilde{E}_{\mathrm{XC}}[n] - E_{\mathrm{XC}}[n]\bigr| & \;\le\; \varepsilon,\\
\left\|\frac{\delta \widetilde{E}_{\mathrm{XC}}}{\delta n}[n] - \frac{\delta E_{\mathrm{XC}}}{\delta n}[n]\right\| & \;\le\; K\,\varepsilon.
\end{align}
That is, the approximate functional and its functional derivative (XC potential) differ from the exact ones by at most $\varepsilon$ and $K\,\varepsilon$, respectively, \emph{uniformly} over $U$.
\end{assumption}

\begin{assumption}[Approximate Iteration in $U$]\label{assump:approxIter}
The modified map $\widetilde{\mathcal{T}}$ also keeps iterates in $U$, i.e.\ $\widetilde{\mathcal{T}}[U] \subset U$, and has a unique fixed point $\widetilde{n} \in U$. 
\end{assumption}

Under these assumptions, we want to show the final SCF solution $\widetilde{n}$ remains $O(\varepsilon)$ close to $n^*$, and the energies differ by $O(\varepsilon)$.
We propose the following theorem:

\begin{theorem}[Iterative SCF Stability]\label{thm:IterativeSCF}
Let $\{n^{(k)}\}$ be the SCF sequence defined by $n^{(k+1)} = \widetilde{\mathcal{T}}[n^{(k)}]$, for $k=0,1,2,\dots$, with initial guess $n^{(0)}\in U$. Assume \ref{assump:contract}--\ref{assump:approxIter}. 
Then for each integer $k\ge 0$, 
\begin{equation}\label{eq:e_kBound}
  \|\,n^{(k)} - n^*\|\;\le\;\kappa^k \|\,n^{(0)} - n^*\|\;+\;\frac{\alpha\,\varepsilon}{\,1-\kappa\,}\;\Bigl(1 - \kappa^k\Bigr),
\end{equation}
for some constant $\alpha>0$. 
Consequently, as $k\to\infty$, the limit $\widetilde{n} = \lim_{k\to\infty} n^{(k)}$ satisfies
\begin{equation}\label{eq:FinalLimitBound}
  \|\,\widetilde{n} - n^*\|\;\le C \varepsilon.
\end{equation}
with a constant $C = \frac{\alpha}{1-\kappa}$.
Furthermore, the total energy difference $|\widetilde{E} - E^*|$ is also
$O(\varepsilon)$:

\[
  |\widetilde{E}-E^*|
  \le C'\,\varepsilon,
\]
where $C' = 1 + \frac{M\,\alpha}{1-\kappa}$ for some constant $M > 0$.
\end{theorem}

The proof, provided in App.~\ref{sec:app-proofs}, applies the contraction mapping argument in combination with uniform bounds on the functional derivatives. 
The constant $\kappa$ depends on the specific SCF implementation (e.g., QEX~\cite{qex2025pasqal}), $\alpha$ reflects the quality of the XC potential modeled by the neural XC functional, and $M$ quantifies the proximity of the final converged density in our implementation to the true electron density within the chosen basis set.
Under the stated assumptions, we ensure that the SCF error is bounded at every iteration, which in turn guarantees stability of the entire SCF loop under noisy conditions.

In practice, these theoretical assumptions can not always be met, for instance in cases of vanishing gap situation, often arise and need use of damping schemes~\cite{woods2019computing}.
If $\kappa$ is close to 1, the factor $(1-\kappa)^{-1}$ in $C$ becomes large, so a small uniform error in neural XC functional might still yield a nontrivial final difference. 
This is consistent with the well-known difficulty of SCF in small-gap or gapless systems~\cite{liu2014convergence}. 
Often, advanced mixing or level-shifting schemes can be used to artificially increase $\kappa$-contractivity~\cite{liu2015analysis}.
Hence, our theoretical assumptions can be met through the following mechanisms:

(i)~\textit{Algorithmic stabilization.}  
To help guarantee Assumption~\ref{assump:contract}, if the iteration is not strictly contractive, convergence can be promoted by standard acceleration schemes, such as density mixing (DIIS~\cite{Pulay1980}, EDIIS~\cite{Kudin2002}), preconditioning~\cite{Kerker1981}, and level shifting~\cite{Saunders1973}, among others.

(ii)~\textit{Regularity and bounded error of the functional.}  
To help guarantee Assumption~\ref{assump:uniform}, when the XC functional is represented by a classical MLP or by a QNN, uniform bounds can be established for the total energy. In the QNN case, the Fourier representation of variational circuits~\cite{schuld2021effect} ensures smoothness and local Lipschitz continuity of the learned functional, thereby controlling error propagation in the SCF cycle; see App.~\ref{sec:app-lipschitz}, Lemma~\ref{lem:qnn-lipschitz-L1} for the proof.

(iii)~\textit{Physically motivated initialization.}  
To help guarantee Assumption~\ref{assump:approxIter},
the SCF cycle can be initialized with a non-interacting electron density at $n^{(0)}$ that yields a physically reasonable starting point, as implemented for example in \textsc{PySCF}~\cite{sun2018pyscf}.  
Such initialization increases the likelihood that the Kohn-Sham map $\mathcal{T}$ begins in a locally contractive regime, although this is not universally guaranteed~\cite{woods2019computing}.
In addition, during the training of neural XC functional, the learning rate can be tuned to prevent the model from driving the SCF procedure outside the contractive region.

These arguments are supported by concrete realizations in refs.~\cite{ksr, kasim2021learning}, where classical neural XC functionals were successfully trained, and, more recently, by the neural XC functional Skala~\cite{luise2025accurate}, which competes with standard functionals.
In addition, refs.~\cite{ko2023stochastic, ko2024ground} provide rigorous studies of stochastic SCF and its convergence, which, although not directly focused on our machine-learning task, are complementary to this work and motivate further investigation of the convergence behavior with respect to $\varepsilon$ when scaling to larger systems.

\begin{figure}[ht]
    \centering
    \includegraphics[width=1.0\columnwidth]{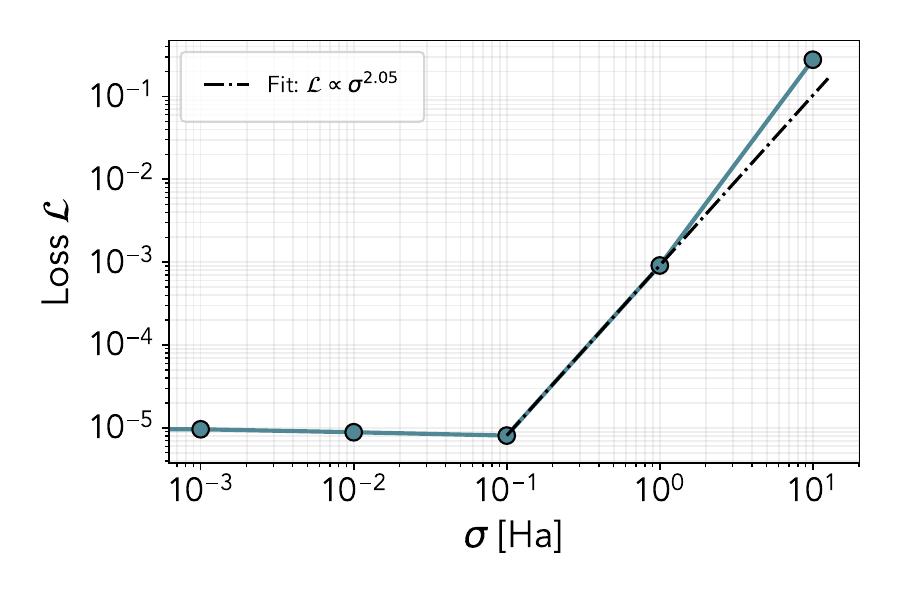}
    \vspace{-0.5cm}
    \caption{
    Dependence of the converged loss (see Eq.~\eqref{eq:loss}) on the noise level $\sigma$ when training the Global KSR model.
    The noise is sampled from a normal distribution $\mathcal{N}(0, \sigma^2)$ with standard deviation $\sigma$ denoting the noise level. 
    The results are obtained within the 1D framework using the same set of molecular geometries and hyperparameters as in Fig.~\ref{fig:h2-1D}. 
    The fit is performed using the last three points ($\sigma \in \{0.1, 1, 10\}$ Ha), and all data points correspond to the mean converged loss over four different random initializations of the model parameters.
    }
    \label{fig:results-noise-ksr}
\end{figure}

The validation of the theory lies in its verification within a concrete implementation. 
To this end, we quantify the effect of typical noise sources on the training capacity of our models. 
For this purpose, we select the most sophisticated classical model used in this work -- Global KSR -- without modifying its hyperparameters from ref.~\cite{ksr}. 
Noise is modeled by sampling from a Gaussian distribution $\mathcal{N}(0, \sigma^2)$, where the standard deviation $\sigma$ specifies the noise level and mimics typical shot noise in QNNs. 
The sampled noise is added to the output of the Global KSR model. 
In this setting, Fig.~\ref{fig:results-noise-ksr} reports the actual errors (loss) produced by our KS-DFT implementation using Global KSR. This experiment isolates the effect of noise and directly tests whether the expected linear scaling is recovered in a full SCF implementation.
We observe that our KS-DFT framework is robust up to significant levels of Gaussian noise ($\sigma = 0.1$ Ha) in the KSR model and, for $\sigma > 0.1$ Ha, the (converged) energy and density squared error (see Eq.~\eqref{eq:loss}) scale with $\mathcal{O}(\sigma^{2.05})$ as expected. 
Hence, the total error $\varepsilon$ scales approximately linearly with $\sigma$, consistent with Theorem~\ref{thm:IterativeSCF}. 
The slight deviation from exact quadratic scaling in the fit can be attributed to the instability in training at very large noise levels ($\sigma = 10$ Ha).
This shows that the Global KSR model, is expressive enough to compensate for the added noise and still converge to a similar loss in magnitude up to $\sigma = 0.1$ Ha.

For an XC functional trained in a noisy setting, an important question is whether it can generalize to unseen data and maintain a similar error. 
We test this hypothesis in Tab.~\ref{tab:noisy-ksr}, where the model interpolates the dissociation curve of hydrogen and is evaluated on unseen data points. 
All reported metrics remain stable for $\sigma < 10^{-1}$, e.g., NPE is of order of $10^{-3}$ as in the noiseless $\sigma = 0$ case, indicating that the noise does not impair the model’s generalization capacity. 
Instead, it appears to act as a regularizer, mitigating overfitting, as also observed for QNNs in ref.~\cite{somogyi2024method}.

To isolate the effect of quantum noise on the quality of the XC functional, we consider a purely quantum model, Product QNN, without any classical post-processing. 
We choose a simple configuration with $N_q = 3$ qubits and $N_l = 3$ layers, using the local XC embedding from Eq.~\eqref{eq:ksr-xc}, identical to that employed in the Global KSR model above and commonly used in KS-DFT implementations. 
Thus, the only change from the previous experiment in Tab.~\ref{tab:noisy-ksr} is the model architecture.
We focus on a standard noise mechanism -- the combination of amplitude damping and phase damping -- described in detail in App.~\ref{sec:noise-channels}. 
Shot noise and gate noise are studied separately:
(i) For gate noise, we fix $N_{\mathrm{shots}} = \infty$ and vary the noise probability $p \in \{10^{-3}, 10^{-2}, 10^{-1}\}$, error rates achievable on current quantum computers based on superconducting qubits~\cite{bravyi2024high}.
(ii) For shot noise, we set $p = 0$ and vary $N_{\mathrm{shots}} \in \{ 10, 10^2, 10^4, 10^6 \}$.
Table~\ref{tab:noisy-lqnn} summarizes the results. 
For all gate-noise experiments, energy and density errors remain of the same order of magnitude regardless of $p$, as the baseline model already exhibits relatively large but constant errors of order $10^{-2}$ Ha. 
In contrast, reducing $N_{\mathrm{shots}}$ below $10^4$ significantly degrades generalization performance, with errors increasing by roughly an order of magnitude.
As shown in App.~\ref{sec:shots}, the typical statistical error scales as
$
\varepsilon = \mathcal{O}\left(N_{\mathrm{shots}}^{-1/2}\right),
$
which is clearly observed for $N_{\mathrm{shots}} = 10^2$, yielding NPE $\approx 10^{-1}$. This demonstrates that within the differentiable KS-DFT framework, at least $N_{\mathrm{shots}} > 10^4$ are required to match the noiseless performance in case of Local PrQNN, and achieving chemical precision of 1 mHa would likely require $N_{\mathrm{shots}} \approx 10^6$ with other, more expressive, models.
Next, we discuss in detail the resource estimates for hardware implementation based on these estimates.

\begin{table}
    \centering
    \caption[]{
    Performance of Global KSR within our differentiable 1D KS-DFT framework for fitting the dissociation profile of the hydrogen molecule, with additive noise drawn from a normal distribution $\mathcal{N}(0, \sigma^2)$.
    The results are obtained using the same hyperparameters and molecular geometries as in Tab.~\ref{tab:results-1D-h2}. 
    }
    \label{tab:noisy-ksr}
    \smaller
    \setlength{\tabcolsep}{6.5pt}
    \begin{tabular*}{\columnwidth}{llll}    \hline\hline
    {$\sigma$ [Ha]} & {$\Delta E$ [Ha]} & Av. MSE($n$) & NPE [Ha] \\
    \hline
0 & $(1.0 \pm 1.0)\cdot10^{-3}$ & $(2.4 \pm 2.4)\cdot10^{-6}$ & $4.0\cdot10^{-3}$ \\
$10^{-3}$ & $(1.4 \pm 1.8)\cdot10^{-3}$ & $(7.5 \pm 5.5)\cdot10^{-6}$ & $2.60\cdot10^{-3}$ \\
$10^{-2}$  & $(2.8 \pm 0.7)\cdot10^{-3}$ & $(7.5 \pm 9.7)\cdot10^{-6}$ & $3.85\cdot10^{-3}$ \\
$10^{-1}$   & $(1.5 \pm 1.0)\cdot10^{-2}$ & $(2.7 \pm 2.1)\cdot10^{-5}$ & $3.17\cdot10^{-2}$ \\
$10^{0}$   & $(5.2 \pm 3.1)\cdot10^{-2}$ & $(8.6 \pm 6.0)\cdot10^{-4}$ & $1.15\cdot10^{-1}$ \\
$10^{1}$  & $(4.9 \pm 0.2)\cdot10^{0}$  & $(6.3 \pm 4.0)\cdot10^{-3}$ & $5.32\cdot10^{0}$ \\
\hline\hline
    \end{tabular*}
\end{table}

\begin{table}
    \centering
    \caption[]{
    Performance of Local PrQNN in our differentiable 1D KS-DFT framework for fitting the dissociation profile of the hydrogen molecule, with gate (amplitude and phase damping channels) and measurement noise.
    The results are obtained using the same hyperparameters and molecular geometries as in Tab.~\ref{tab:results-1D-h2}. 
    }
    \label{tab:noisy-lqnn}
    \smaller
    \setlength{\tabcolsep}{3.5pt}
    \begin{tabular*}{\columnwidth}{lllll} 
    \hline\hline
    {$p$} & {$N_{\mathrm{shots}}$} & {$\Delta E$ [Ha]} & Av. MSE($n$) & NPE [Ha] \\
    \hline
$0$ & $\infty$ & $(1.6 \pm 1.6)\cdot10^{-2}$ & $(1.5 \pm 1.5)\cdot10^{-4}$ & $5.35\cdot10^{-2}$ \\
$10^{-3}$ & $\infty$  & $(1.7 \pm 1.5)\cdot10^{-2}$ & $(2.3 \pm 2.4)\cdot10^{-4}$ & $5.27\cdot10^{-2}$ \\
$10^{-2}$ & $\infty$  & $(1.6 \pm 1.6)\cdot10^{-2}$ & $(1.5 \pm 1.5)\cdot10^{-4}$ & $5.35\cdot10^{-2}$ \\
$10^{-1}$ & $\infty$  & $(1.7 \pm 1.5)\cdot10^{-2}$ & $(1.5 \pm 1.4)\cdot10^{-4}$ & $5.27\cdot10^{-2}$ \\
$0$ & $10^{6}$ & $(1.6 \pm 1.6)\cdot10^{-2}$ & $(1.5 \pm 1.5)\cdot10^{-4}$ & $5.36\cdot10^{-2}$ \\
$0$ & $10^{4}$ & $(1.6 \pm 1.6)\cdot10^{-2}$ & $(1.5 \pm 1.5)\cdot10^{-4}$ & $5.42\cdot10^{-2}$ \\
$0$ & $10^{2}$ & $(5.5 \pm 11.9)\cdot10^{-2}$ & $(1.6 \pm 4.1)\cdot10^{-3}$ & $3.92\cdot10^{-1}$ \\
$0$ & $10^{1}$ & $(1.3 \pm 1.7)\cdot10^{-1}$ & $(4.0 \pm 6.3)\cdot10^{-3}$ & $5.02\cdot10^{-1}$ \\
    \hline\hline
    \end{tabular*}
\end{table}

\subsection{Resource estimations for quantum hardware implementation}

A robust evaluation of a competitive quantum model ideally involves assessing its time-to-solution, generalization capacity, parameter efficiency, and expressive power. 
In this context, the resources required for evaluating quantum XC models align with those conventionally utilized for DQC-based QNNs.~\cite{kyriienko2021solving}
We focus on the time-to-solution, which is contingent on the total number of measurements (shots) required for an accurate output evaluation of quantum models.
To that end, we examine the number of shots $N^{\text{tot}}_{\text{shots}}$ needed for the gradients $\nabla_{\theta}E_{\mathrm{XC}}$, $v_{\mathrm{XC}}[n](\mathbf{r})=$ $\delta E_{\mathrm{XC}}[n] / \delta n(\mathbf{r})$ and energy $E_\mathrm{XC}$ evaluations during the training phase.

\textit{Assumptions.} To reasonably assess the resource requirements, it is important to note that, to achieve a desired precision $\varepsilon$ on the total energy, the required number of epochs $N_{\text{epochs}}$ and the number of model parameters $N_{\text{param}}$ are likely to depend on the problem size, i.e., the number of atoms $N_{\mathrm{atoms}}$. 
Larger and more diverse datasets may require a more expressive XC model -- thus a greater number of parameters -- as well as adjustments to the batch size and, consequently, an increase in the number of epochs.
In such situations, with large $N_{\text{param}}$ (i.e., large depth), the barren plateau problem may arise, necessitating the use of barren-plateau-resistant architectures~\cite{schatzki2022theoretical} and appropriate mitigation techniques~\cite{sannia2024engineered}.
For the sake of providing an approximate estimate, we assume that $N_{\text{epochs}}$ and $N_{\text{param}}$ do not depend on the problem size.
The exact scaling would depend on the chosen model and other hyperparameters of the framework.
Since neural XC functionals in our approach can be designed to be independent of the problem size, they do not need to change as $N_{\mathrm{atoms}}$ increases, although they may no longer achieve accuracy $\varepsilon$ without architectural adjustment.
Thus, our estimates do not guarantee achieving accuracy $\varepsilon$ but instead provide an approximate minimal shot budget under the aforementioned assumptions.

\textit{Energy evaluation.} For a single evaluation of the XC functional, the cost operator $\hat{C} = \sum_{i=1}^{N_\mathrm{q}} \hat{Z_i},$ in the loss function is diagonal by design. 
It allows for the simultaneous measurement of all its terms (due to mutual commutation), thereby mitigating associated errors on noisy quantum processors.
This drastically reduces the number of shots in comparison to common \textit{ab intio} variational approaches~\cite{peruzzo2014variational, mcardle2019variational} that require the cost operator to be the electronic Hamiltonian of given system, which typically has many non-commuting terms to be measured (i.e., $\mathcal{O}(N^{4})$ terms, where $N$ is the system size).
In App.~\ref{sec:shots}, we establish that $N_{\mathrm{shots}} = \mathcal{O}(\varepsilon^{-2})$ measurements per expectation value are required to achieve an accuracy $\varepsilon$ on the XC energy $E_{\mathrm{XC}}$ when using total magnetization cost operator. 

\textit{Gradient evaluation.} To evaluate $\nabla_{\theta}E_{\mathrm{XC}}$ or $v_{\mathrm{XC}}$ by using the parameter-shift rule~\cite{schuld2019evaluating}, two expectation values need to be computed at each epoch during the training (see Appendix~\ref{app:diff}).
In practical terms, considering that we have $N_{\text{epochs}}$ during training, we require $N_{\text{shots}}$ of order $\mathcal{O}(\varepsilon^{-2})$ per expectation value to achieve the desired accuracy $\varepsilon$. 
Consequently, the total number of shots during training performed for  $\nabla_{\theta}E_{\mathrm{XC}}$, taking into account \(N_{\text{param}}\) training parameters, amounts to $2N_{\text{shots}}N_{\text{param}}N_{\text{epochs}}$.
Similarly for $v_{\mathrm{XC}}$, the gradient $\delta E_{\mathrm{XC}}[n] / \delta n(\mathbf{r})$ is evaluated at $N_{\mathrm{grid}}$ points at $N_{\text{KS}}$ iterations, resulting in $2N_{\text{shots}}N_{\text{KS}}N_{\text{grid}}N_{\text{epochs}}$ shots.

\textit{Global embedding.} In addition, $E_\mathrm{XC}$ needs to be inferred at every KS iteration.
In our experience, the number of KS iterations can be set at \(N_{\text{KS}} \in [15 - 100]\) depending on the size and complexity (strongly correlated regime) of the system.
See ref.~\cite{das2023accelerating}, where large-scale KS-DFT calculations required up to 30 KS iterations.
Therefore, the total cost in terms of measurements in our framework is given by $$
N^{\text{tot,g}}_{\text{shots}} = (2N_{\text{param}}+2N_{\text{grid}}N_{\text{KS}}+N_{\text{KS}})N_{\text{shots}}N_{\text{epochs}}.
$$
For the scaling approximation, we consider the number of trainable parameters $N_{\text{param}}$ and the number of grid points $N_{\text{grid}}$ to be significantly larger than $N_{\text{KS}}$.
In the following, we consider $N_{\text{KS}}$ constant. 
Hence, the overall cost in terms of measurements is generally expected to scale as 
$$
N^{\text{tot,g}}_{\text{shots}} \sim \mathcal{O}((N_{\text{param}} + N_{\text{grid}})N_{\text{shots}}N_{\text{epochs}}).
$$
This scaling is valid for the global energy embedding (Eq.~\eqref{eq:gl-xc}), where a singular energy evaluation suffices~(\(E_{\text{XC}}[n] \approx F_{\theta}[n]\)).

\textit{Local embedding.} In a local embedding scenario, the XC energy \(E_{\text{XC}}\) is approximated through \(N_{\text{grid}}\) model evaluations, where the total error is the weighted sum of individual errors, leading to number of shots much greater than with the global energy embedding, $N_{\text{grid}}N_{\text{KS}}N_{\text{shots}}N_{\text{epochs}} \gg N_{\text{KS}}N_{\text{shots}}N_{\text{epochs}}$.
The total number of shots required is 
$$
N^{\text{tot,l}}_{\text{shots}} = (2N_{\text{param}}+ 2N_{\text{grid}}N_{\text{KS}} + N_{\text{grid}}N_{\text{KS}})N_{\text{shots}}N_{\text{epochs}}.
$$
Hence, the overall total measurement scaling for local embedding can be estimated as 
$$N^{\text{tot,l}}_{\text{shots}} \sim \mathcal{O}((N_{\text{param}}+ N_{\text{grid}})N_{\text{shots}}N_{\text{epochs}}).$$

\begin{table}[ht]
\centering
\caption[]{Scaling of the number of measurements for the training with the proposed local (Eq.~\eqref{eq:ksr-xc}) and global (Eq.~\eqref{eq:gl-xc}) energy embeddings in quantum-enhanced KS-DFT.
We assume that $N_{\text{KS}}$ remains constant~\cite{das2023accelerating}, and define $N^{\text{grid}}_{\text{param}} = N_{\text{grid}} + N_{\text{param}}$.
} 
\label{tab:scaling}
\smaller
\setlength{\tabcolsep}{4.5pt}
\begin{tabular*}{\columnwidth}{lll} 
\hline\hline
& Local & Global \\
\hline
$E_{\mathrm{XC}}$  &  \(\mathcal{O}(N_{\text{grid}}N_{\text{shots}}N_{\text{epochs}})\)  &  \(\mathcal{O}(N_{\text{shots}} N_{\text{epochs}})\)  \\
$\nabla_{\theta}E_{\mathrm{XC}}$  &  \(\mathcal{O}(N_{\text{param}}N_{\text{shots}}N_{\text{epochs}})\)  &  \(\mathcal{O}(N_{\text{param}}N_{\text{shots}}N_{\text{epochs}})\) \\
$v_{\mathrm{XC}}$   & \(\mathcal{O}(N_{\text{grid}}N_{\text{shots}}N_{\text{epochs}})\)   & \(\mathcal{O}(N_{\text{grid}}N_{\text{shots}}N_{\text{epochs}})\) \\
\hline
Total & $\mathcal{O}(N^{\text{grid}}_{\text{param}}N_{\text{shots}}N_{\text{epochs}})$ & $\mathcal{O}(N^{\text{grid}}_{\text{param}}N_{\text{shots}}N_{\text{epochs}})$  \\
\hline\hline
\end{tabular*}
\end{table}

In addition, the grid size $N_{\text{grid}}$ scales linearly with the number of atoms $N_{\text{grid}} \sim \mathcal{O}(N_{\text{atoms}})$~\cite{becke1988multicenter}.
Thus, under our assumptions, the total cost for training with local or global energy embedding scales linearly with the system size in terms of shots due to the need to evaluate the local XC potential on the grid.
The aforementioned costs are summarized in Table~\ref{tab:scaling}.

\textit{Concrete estimations.}
To estimate the shot budget required for an actual implementation on quantum hardware, we assume the typical values used in this work: $N_{\text{epochs}} = 10^{3}$, $N_{\text{KS}} = 10$, $N_{\text{grid}} = 10^{3}$, and $N_{\text{param}} = 10^{2}$ (e.g., Global QNN).
For the $N_{\mathrm{shots}}$ shots per energy evaluation, Sec.~\ref{sec:proofs} shows, for a QNN-based functional, that deterioration from the noiseless result occurs as soon as $N_{\mathrm{shots}} < 10^{4}$, consistent with the derivations in App.~\ref{sec:shots}. Based on these estimates, achieving $\varepsilon \leq 1$ mHa (chemical precision) is expected to require at least $N_{\mathrm{shots}} \approx 10^{4} - 10^{6}$, since the error scales as $\varepsilon = \mathcal{O}(N_{\text{shots}}^{-1/2})$.
With the global embedding, the total number of shots is
$
N^{\text{tot,g}}_{\text{shots}}.
$
Hence, per epoch, we require $\approx 10^{8} - 10^{10}$ shots, and for the entire training, $N^{\text{tot,g}}_{\text{shots}} \approx 10^{11} - 10^{13}$ shots.

The above assumes sequential measurements; however, only the epochs and KS iterations are inherently sequential. The gradients $\nabla_{\theta} E_{\mathrm{XC}}$ and $v_{\mathrm{XC}}$ can be fully parallelized on quantum hardware (i.e., we set $N_{\text{param}} = 1$ and $N_{\text{grid}} = 1$), reducing the shot requirement to
$$
N^{\text{tot,g}}_{\text{shots}} \approx N_{\text{KS}} N_{\text{shots}} N_{\text{epochs}},
$$
which yields $10^{8} - 10^{10}$ shots for the entire training.
On quantum hardware, the overall time required for training quantum models depends on the single-shot rate, highly dependent on the hardware type.
This consideration ultimately determines the computational resources needed to attain a practical advantage (e.g., lower inference time) over classical XC models.
Assuming an optimistic measurement time on a superconducting device of 1 ms per shot (e.g., $\approx 1$ kHz shot rate reported in ref.~\cite{kim2023evidence}) and $10^{8}$ total shots, complete training could, in principle, be done in about one day. 
Additional measurement schemes, such as neural network estimators~\cite{torlai2020precise}, could be employed to reduce variance, which would otherwise prevent achieving chemical precision.

Given the imperative to consistently achieve chemical precision for meaningful outcomes (e.g., accurate energy barriers), global embedding of the quantum XC functional (see Eq.~\eqref{eq:gl-xc}) emerges as the most promising near-term approach (i.e., especially in situations with small grid size, $N_{\text{grid}} \ll N_{\text{shots}}$ and $N_{\text{grid}} \ll N_{\text{epochs}}$).
In our experiments, as our loss function is defined against energies, the fact that the global energy embedding does not scale with grid size significantly reduces the energy error of our models in comparison to the local embedding, leading to the best outcomes.

\section{\label{sec:conclusion} Conclusions and outlook}

In this work, we developed various (Q)NN-based architectures and evaluated their capabilities to represent XC functionals within KS-DFT. 
To that end, we extended existing differentiable KS-DFT codes based on JAX and PySCF frameworks, in both 1D and 3D settings, to accommodate quantum models.
Next, we defined four classes of models capable of incorporating an arbitrary number of quantum and/or classical layers.
As in other works,~\cite{ksr, kasim2021learning} we classified the models by their locality (i.e., the fraction of density points considered as the input). 
To integrate these architectures into KS-DFT, we introduced a global embedding strategy, well suited for quantum architectures due to its minimal measurement requirements on quantum hardware.
In terms of scalability, both local and global approaches are influenced by the system size due to the need to compute the local XC potential.
However, a significant advantage of our model is that the global energy embedding does not scale with the grid size. 
This feature substantially improved the accuracy of our model, yielding the best results in our experiments.
Finally, we performed numerical simulations of the entire training process on a classical computer.

The results of our experiments demonstrate the promising capabilities of QNNs in efficiently modelling the XC energy functionals.
In our proof-of-concept evaluation, we showed that quantum models can achieve performance comparable to classical models with a similar number of parameters (Sec.~\ref{sec:qnn-vs-mlp}), offering small improvements in generalization.
Moreover, the potentially greater expressivity of QNNs compared to classical models~\cite{abbas2021power} (e.g., exponential number of Fourier frequencies~\cite{schuld2021effect} as discussed in Sec.~\ref{sec:feature_maps}) may not be strictly necessary for approximating the exact XC functional.
Hence, the metrics typically adopted in analyzing quantum ML models (such as expressivity and entangling capability) may not be directly related to the ability of those models to perform well on a specific classical dataset.
This observation motivates further investigations on larger datasets and more complex instances of quantum models, a non-trivial task given the inherent challenges of simulating quantum mechanics on classical computers.

Only fully global (quantum-enhanced) models, where the whole density is considered, produced the best results.
In particular, Global QNN provided the simplest, yet sufficiently expressive architecture to model with chemical precision the small molecular systems considered in this work, both in 1D (H$_2$) and 3D (planar H$_4$), even in the presence of strong correlation.
With a mere 316 trainable parameters, Global QNN produced similar results to the state-of-the-art Global KSR model, which requires ten times more parameters.
The same modest model demonstrated satisfactory generalization capability on the H$_2$-H$_2$ system (two separate dihydrogen molecules), despite being trained solely on H$_2$ and H$_4$.
This offers a promising perspective on the potential of QNNs for generalization within this framework.

Moreover, we incorporated a fractional orbital occupation approach into our framework, which can be also employed to train purely classical models.
This FO approach was instrumental in enhancing the accuracy of our models in 3D. 
By employing the FO method, we reduced both the average energy error and its variance by approximately six-fold. 
In that way, we effectively mitigated energy oscillations and substantially improved the extrapolation accuracy of Global QiCNN.

To ensure the stability of our framework in realistic noisy conditions present in current quantum hardware -- specifically, the SCF loop, which is the core subroutine of KS-DFT -- we developed a theoretical framework in Sec.~\ref{sec:proofs} that establishes error bounds and specifies the conditions the model and implementation must satisfy to guarantee well-controlled errors.
We verified this theory in a concrete implementation under various noise settings, including typical hardware noise, and recovered the predicted linear scaling.
Furthermore, we provided resource estimates indicating the number of measurements required for a hardware implementation, demonstrating the feasibility of our approach when observables are measured in parallel.

Additionally, in our QML approach, a quantum-enhanced neural functional can be initially trained and subsequently fine-tuned with additional data as needed. 
Once trained, this functional can be directly applied to any system to infer its properties. 
This stands in contrast to variational quantum eigensolver~\cite{peruzzo2014variational} and imaginary time~\cite{mcardle2019variational} approaches, which necessitate full optimization for each system.
The latter methods evolve the quantum state using the exact electronic Hamiltonian. 
In contrast, QML methods depend on the quality of the training data to ensure qualitatively correct predictions.

As an outlook, in this study, we provided an initial evaluation of proposed techniques on quantum and quantum-inspired neural architectures and molecules.
In future works, it would be necessary to expand our training dataset substantially to understand the practical efficiency of quantum-enhanced models. 
Efforts to construct such datasets are currently ongoing.~\cite{casares2023graddft}
The capability of our models to generalize to significantly larger systems is a critical aspect to investigate. 
Specifically, can these models effectively learn from small, strongly correlated systems and subsequently generalize accurately to larger systems?

Additionally, the FO approach~\cite{frac-occ} employed in this study has been evaluated solely on dimers. 
It remains to be verified whether a single value of the temperature parameter $\gamma$, as detailed in Appendix~\ref{app:frac-occ}, is adequate for training on larger datasets.
Alternatively, parameter-free FO approaches, such as the quadratically convergent method can be employed instead~\cite{cances2003quadratically}.
In this work, we opted for the simplest-to-implement method.

An essential direction for future work is to establish rigorous criteria under which QNNs provide genuine advantages over classical models, moving beyond general expressivity arguments.
While recent advances in ML-based XC design~\cite{ksr,kasim2021learning,luise2025accurate} have yielded promising results, the quest for a universal functional remains unresolved. 
To systematically explore this challenge, we have released QEX~\cite{qex2025pasqal}, a flexible framework that supports arbitrary classical-quantum hybrids.
Our approach is inherently compatible with barren plateau mitigation techniques (e.g., engineered dissipation~\cite{sannia2024engineered}) and with architectures naturally resilient to trainability issues (e.g., QCNNs~\cite{cong2019quantum} and permutation-equivariant models~\cite{schatzki2022theoretical}). 
Looking ahead, inspired by quantum-enhanced DFT approaches that leverage QPU-prepared ground states~\cite{sheridan2024enhancing,koridon2025learning}, we propose using QNNs as representation learners that extract task-relevant features from approximate ground-state wavefunctions. 
Analogous to QCNNs on quantum data~\cite{cong2019quantum}, such models could provide inductive biases tailored to electronic-structure tasks. 
This perspective positions QNNs not merely as end-to-end regressors for the XC functional, but as modules within a hybrid pipeline for learning more expressive and accurate approximations. 

The advent of early fault-tolerant quantum computing (FTQC) architectures is expected to significantly broaden the range of feasible quantum algorithms for chemistry. 
While fully fledged protocols such as Quantum Phase Estimation (QPE) will ultimately become the method of choice as devices mature, recent studies already show that logical-qubit architectures can outperform variational methods on near-term problems~\cite{van2024end}. 
Such approaches are fully compatible with our framework and may serve as a natural stepping stone prior to large-scale QPE. 
In the broader landscape, we anticipate that FTQC-based methods and DFT will coexist as complementary paradigms, with hybrid strategies likely to play a central role in the foreseeable future.

Overall, our findings emphasize the promising potential of QNN-based XC functionals within DFT. 
With this contribution, we provided an initial blueprint that would serve as a stepping stone towards the development of more optimized frameworks and quantum-enhanced neural architectures that can tackle even more challenging chemical problems with greater accuracy and efficiency.

\section{\label{sec:acknowledgements}Acknowledgements}

The authors thank the Pasqal team for fruitful discussions, especially Atiyo Ghosh, Andrea A. Gentile, John West and Louis-Paul Henry for their valuable insights.

\appendix

\section{Incorporating fractional occupation in RKS-DFT~\label{app:frac-occ}}

Chai's algorithm for fractional occupation is based on the argument that the distribution of occupation numbers is independent of the specific many-body system and adheres to a universal form described by the Fermi-Dirac distribution.~\cite{frac-occ}
Hence, the density can be expressed as 
$n({\bf r}) =  \sum_{i=1}^{\infty} g_{i} |\psi_{i}({\bf r}) |^{2},$
where the occupation number $g_{i}$ is the Fermi-Dirac function 
$ g_{i} = \{1+\text{exp}[( \epsilon_{i} - \mu)/ \gamma] \}^{-1},$
which obeys the following two conditions, 
$\sum_{i=1}^{\infty} g_{i} = N,$ and $ 0\le g_{i} \le 1,$ where $\epsilon_{i}$ is the orbital energy of the $i^{\text{th}}$ orbital $\psi_{i}({\bf r})$, and $\mu$ is the chemical potential chosen to maintain a constant number of electrons $N$. 
To account for this modification, the expression for the energy is then given by
\begin{equation}
\begin{aligned}
E[n] = & T_s^\gamma\left[\left\{g_i, \psi_i\right\}\right]-\frac{\gamma}{k_B}S\left[\left\{g_i\right\}\right]+\int n(\mathbf{r}) v(\mathbf{r}) d \mathbf{r} \\
& +E_{\mathrm{H}}[n]+E^{\gamma}_{\text{XC}}[n].
\end{aligned}
\end{equation}
with $\gamma$-dependent XC energy given by $E^{\gamma}_{\text{XC}}[n] = E_{\text{XC}}[n]+E_\gamma[n]$, where $E_\gamma[n]$ represents the difference between the non-interacting kinetic free energy at zero temperature and that at temperature $\gamma$.
The entropy contribution is given by $-\frac{\gamma}{k_B} S_s^\gamma\left[\left\{g_i\right\}\right]=\gamma \sum_{i=1}^{\infty}\left[g_i \ln \left(g_i\right)+\left(1-g_i\right) \ln \left(1-g_i\right)\right]$.
Given a suitable electronic temperature $\gamma$, the chemical potential $\mu$ is fixed by $\sum_{i=1}^{\infty}\left\{1+\exp \left[\left(\epsilon_i-\mu\right) / \gamma\right]\right\}^{-1}=N$ to conserve the number of particles.
In essence, only the entropy contribution needs to be added to an existing RKS-DFT implementation and $E^{\gamma}_{\text{XC}}[n]$ is learned by the neural XC.
This approach does not impact the overall computational scaling. 
The trade-off for improved density representability is the inclusion of a single free parameter, the temperature $\gamma$, which may require initial tuning for consistency across systems, a subject that warrants further investigation.

\section{Integration grids and neural XC functionals \label{app:grids}}

\begin{figure*}
    \centering
    \includegraphics[width=2.\columnwidth]{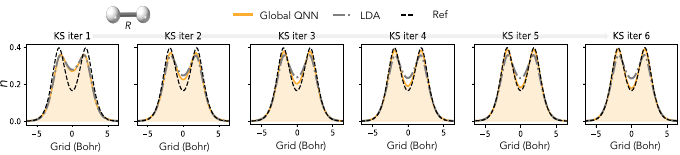}
    \caption{
    Electron density of a hydrogen molecule ($R = 3.84$ Bohr) at first six KS iterations using Global QNN within the 1D KS-DFT framework. 
    At every KS iteration, the electron density is improved, eventually reaching the reference one.
    The reference (Ref) result is generated using DMRG by Li~\etal~\cite{ksr}.
    }
    \label{fig:den-ks-iter}
\end{figure*}

A neural XC model needs to implicitly learn a function $f_{\theta}$ that is adapted to the $\{w_i\}^{N_{\mathrm{grid}}}_{i=1}$ grid points, distributed according to a fixed grid type.
There are multiple choices of grid types and densities (number of grid points $N_{\mathrm{grid}}$), with a finer grid mesh yielding more accurate results at a higher computational cost.
Ideally, a neural XC model needs to be trained/constrained to be invariant to finer mesh grids of the same type.
See Fig.~\ref{fig:den-ks-iter} where a trained Global QNN model is used to compute the electron density throughout KS iterations on a fixed grid of 513 points.
This adds to the challenge of making generalizable neural XC functionals where grid (hence the number of model inputs) needs to be adapted and neural XC functional made consistent with the changes.
For instance, a potential solution is to design a grid invariant embedding or add constraints in the loss to make the NN $f_{\theta}$ invariant to grid density changes.
Another approach would be to use also the coordinates $\{ \mathbf{r}_i \}_{i=1}^{N_{\mathrm{grid}}}$ as inputs of the model $f_{\theta}$ to explicitly include the information about the distribution of grid points in the Cartesian space.
However, this approach would quadruple the number of inputs, so it is not currently considered for simplicity.
In our work, we minimize grid variations by using a single grid type (e.g., of Lebedev type\cite{Lebedev-USSR-1975}) and the same number of grid points for all molecules whenever possible (our models have a fixed number of $N_{\mathrm{grid}}$ inputs) in order to facilitate the training.

\section{Building blocks of quantum neural networks \label{app:qnn}}

In this section, we present typical differentiable quantum circuits (DQCs) that are used in this work.

\subsection{Quantum circuit ansatz}

As an ansatz $U_{a}$ layer, we employ the popular Hardware-Efficient Ansatz (HEA) introduced in the context of VQE for chemistry.~\cite{peruzzo2014variational}
The heuristic ansatz structure is not specifically tailored to the problem; thus, it represents the most general quantum model, closely resembling a classical MLP.
It is composed of selected single qubit rotations (i.e., $R_{z}(\theta)= e^{-\frac{i}{2}\theta \hat{Z}}$) that are placed on all available qubits with distinct angles $\theta_a$,
\begin{equation}
    U_{\mathrm{sqr}}(\theta_a)= \prod^{N_q}_{i=1} e^{-\frac{i}{2}\theta^{i}_{a,x}\hat{Z}}e^{-\frac{i}{2}\theta^{i}_{a,y}\hat{X}}e^{-\frac{i}{2}\theta^{i}_{a,z}\hat{Z}},
\end{equation}
which enables the representation of the most general rotation on a Bloch sphere.
The entanglement block consists of an arbitrary pattern of (parameterized) multi-qubit gates that introduce interactions between qubits. 
For example, considering the two-qubit case, a unitary is given by
\begin{equation}
     U_{\mathrm{ent}}=\prod_{(i, j) \in S} \operatorname{CNOT}(i, j)
\end{equation}
where $\operatorname{CNOT}(i, j) = e^{-\frac{i}{4}{(I_i - \hat{Z}_i)(I_j-\hat{X}_j)}}$ and $S$ represents an ordered set of pairs of qubit indices. 
We will employ an ``alternate ladder'' set $S_{\mathrm{al}} = S_1 \cup S_2 $ with $S_{1} = \left((0,1), (2,3),..., (N_q-1, N_q)\right)$ and $S_{2} = \left((1,2), (3, 4), ..., (N_q-2, N_q-1)\right)$ unless specified otherwise.
Note that a parametrized entanglement layer can also be defined as
$U_{\mathrm{ent}}^{\mathrm{RZZ}}(\theta_{\mathrm{ent}})=\prod_{(i, j) \in S} R_{\mathrm{ZZ}}^{i, j}\left(\theta_{\mathrm{ent}}^{i, j}\right)=\prod_{(i, j) \in S} \operatorname{CNOT}(i, j) e^{-\frac{i}{2}\theta_{\mathrm{ent}}^{i, j}\hat{Z}}\operatorname{CNOT}(i, j),$
where $R_{\mathrm{ZZ}}^{i, j}$ represents the parametrized $\mathrm{ZZ}$ gate between qubits.
This allows the model to cancel the interaction $\left(\theta_{ent}^{i, j}=0\right)$ between qubits if necessary.  
The final quantum state $|\psi\rangle$ prepared by repetitive applications of the feature map and ansatz is given by
\begin{equation}
|\psi (n, \theta)\rangle=\prod_{l=0}^{N_l}\left(\prod_{d=0}^{N_d} {U}^{l,d}_{\mathrm{ent}}{U}^{l,d}_{\mathrm{sqr}}({\theta^{l,d}_{a}}){{U}}^{l}_{f}(n)\right)|0\rangle,
\label{eq:state}
\end{equation}
and subsequently used in Eq.~\eqref{eq:qmodel} as $f_{{\theta}}(n)=\langle \psi(n, \theta) | \hat{C} | \psi(n, \theta) \rangle$.
Throughout this work, all the quantum layers will be evaluated in this manner.
Next, we discuss how classical data can be embedded in quantum circuits ${U}_{f}(n)$.

\section{Differentiation of quantum models \label{app:diff}}

Gradient evaluation is an essential component for training NNs in general.
Since (quantum) NNs can be thought of as a function $f_{\theta}$, there are various approaches for differentiation.
The standard approach in classical ML is the \textit{Automatic Differentiation} (AD) that exploits the chain rule and stores the intermediate values of subfunctions (layers) of ML models to evaluate the gradients of $f_{\theta}$.
This approach has an advantageous scaling that enables training large classical models and, hence, is always used for the differentiation of classical layers in this work.
In the context of quantum computing, there has been a need to develop an efficient approach, beyond finite differences, for gradient evaluation.
Schuld~\etal~showed that due to the trigonometric relations within the unitaries associated with single qubit rotations, $U_{\mathrm{sqr}} = e^{-\frac{i}{2}\theta \hat{O}}$ (that include single-qubit Pauli generators $\hat{O}$), it is possible to derive a simple and exact derivative formula which highly resembles finite differences, the so-called \textit{Parameter-Shift Rule} (PSR):~\cite{schuld2019evaluating}
\begin{equation}
    \nabla_\theta f_{\theta} (n) = \frac{f_{\theta+\pi / 2}(n)-f_{\theta-\pi / 2}(n)}{2}.
\end{equation}
The generalized PSR (GPSR) rule~\cite{kyriienko2021generalized} is needed for general multi-qubit generators.

\section{Number of measurements}
\label{sec:shots}

To determine the measurement budget to reach an accuracy $\varepsilon$ with QNNs presented in this work, we follow the standard frequentist approach~\cite{wecker2015progress, romero2018strategies}. 
The output of QNNs is obtained by measuring the total magnetization cost operator, which can be expressed as the sum of the expectation values of Pauli-\( Z \) operators over all $N_{{q}}$ qubits,
$
M = \sum_{i=1}^{N_{{q}}} \langle \hat{Z}_i \rangle.
$
Assume each term \( \hat Z_i \) is measured \( m_i \) times. 
The statistical error on each term, \( \varepsilon_i \), is given by
$
\varepsilon_i^2 = \frac{\text{Var}[\langle \hat Z_i \rangle]}{m_i}.
$
Since \( \langle \hat Z_i \rangle \in [-1, 1] \), its variance is upper bounded by 1,
$
\text{Var}[\langle \hat Z_i \rangle] \leq 1.
$
To achieve a total error \( \varepsilon \), we can distribute the per-term errors such that
$
\varepsilon_i^2 = \frac{1}{N_q} \varepsilon^2.
$
Then, the total number of measurements $N_{\mathrm{shots}}$ across all qubits becomes

\begin{equation}
N_{\mathrm{shots}} = \sum_{i=1}^{N_{{q}}} m_i = \sum_{i=1}^{N_{{q}}} \frac{\text{Var}[\langle \hat Z_i \rangle]}{\varepsilon_i^2}
\leq \sum_{i=1}^{N_{{q}}} \frac{1}{\varepsilon^2 / N_q}
= \frac{N_{{q}}^2}{\varepsilon^2}.
\end{equation}
Since all the operators are diagonal, they can be measured simultaneously, resulting in a total shot cost of $N_{\mathrm{shots}} = \mathcal{O}(\varepsilon^{-2})$.

\section{Noise channels} \label{sec:noise-channels}

In Sec.~\ref{sec:shots}, we discussed the error arising from finite sampling. 
In modeling realistic quantum devices, additional noise processes are commonly represented by completely positive trace-preserving (CPTP) maps acting on a qubit’s density matrix $\rho = |\Psi\rangle\langle\Psi|$.
Any CPTP map can be expressed in terms of a set of Kraus operators $\{ \hat K_i\}$ satisfying $\sum_i \hat K_i^\dagger \hat K_i = \mathbb{I}$. Two common single-qubit noise channels are \emph{dephasing} and \emph{amplitude damping}, which we focused on in our experiment in Tab.~\ref{tab:noisy-lqnn} implemented in QEX~\cite{qex2025pasqal}.

\textit{Dephasing (phase-damping) channel.}
The dephasing channel describes the decay of quantum coherences without altering computational-basis populations. A Kraus representation is
\begin{equation}
    \hat K_0 = \sqrt{1 - p}\,\mathbb{I}, 
    \quad 
    \hat K_1 = \sqrt{p}\,\hat Z ,
\end{equation}
where $0 \le p \le 1$ is the dephasing probability and $\hat Z$ is the Pauli-$Z$ operator. 
The channel acts as
\begin{equation}
    \mathcal{E}_{\mathrm{deph}}(\rho) 
    = \hat K_0 \rho \hat K_0^\dagger + \hat K_1 \rho \hat K_1^\dagger.
\end{equation}

\textit{Amplitude-damping channel.}
The amplitude-damping channel models energy relaxation from $\lvert 1 \rangle$ to $\lvert 0 \rangle$, as in spontaneous emission. 
A standard Kraus representation is
\begin{equation}
   \hat  K_2 =
    \begin{pmatrix}
    1 & 0 \\
    0 & \sqrt{1-p}
    \end{pmatrix},
    \quad
   \hat  K_3 =
    \begin{pmatrix}
    0 & \sqrt{p} \\
    0 & 0
    \end{pmatrix} ,
\end{equation}
where $p$ is the probability of decay during the noise process. 
The map is defined as
\begin{equation}
    \mathcal{E}_{\mathrm{amp}}(\rho) 
    = \hat K_2 \rho \hat K_2^\dagger + \hat K_3 \rho \hat K_3^\dagger.
\end{equation}
In Sec.~\ref{sec:proofs}, we analyze the performance of a typical QNN in our framework by applying these noise channels and tuning their strength via the parameter $p$.

\section{Proofs of error bounds} 
\label{sec:app-proofs}

\begin{proof}
We prove Theorem~\ref{thm:IterativeSCF}  under assumptions \ref{assump:contract}--\ref{assump:approxIter}, as stated in Sec.~\ref{sec:proofs}. 
We aim to show that, under an error $\varepsilon$, the exact self-consistent density $n^{\ast}$ obtained from the exact SCF map $\mathcal{T}$ (with the exact XC functional) differs from the approximate self-consistent density $\tilde{n}$ obtained from the approximate SCF map $\widetilde{\mathcal{T}}$ by at most $\mathcal{O}(\varepsilon)$.

\textit{Step 1:} We define iteration errors. 
Let $e_k = \|\,n^{(k)} - n^*\|$, with $||\cdot|| \equiv ||\cdot||_2$. 
Then by definition of $\mathcal{T}$,
\[
  n^{(k+1)} - n^*
  \;=\;
  \widetilde{\mathcal{T}}[n^{(k)}] \;-\; \mathcal{T}[n^*].
\]
We split and add $\mathcal{T}[n^{(k)}]$:
\[
  n^{(k+1)} - n^*
  = \bigl(\widetilde{\mathcal{T}}-\mathcal{T}\bigr)\!\bigl[n^{(k)}\bigr] + \Bigl(\mathcal{T}[n^{(k)}] - \mathcal{T}[n^*]\Bigr).
\]
Hence,
\begin{align}
  e_{k+1}
  &= \bigl\|\,n^{(k+1)} - n^*\bigr\|
   \le \bigl\|\bigl(\widetilde{\mathcal{T}}-\mathcal{T}\bigr)[n^{(k)}]\bigr\|\nonumber \\
   &\quad + \bigl\|\mathcal{T}[n^{(k)}] - \mathcal{T}[n^*]\bigr\|\nonumber\\
  &\le \underbrace{\|\widetilde{\mathcal{T}}[n^{(k)}] - \mathcal{T}[n^{(k)}]\|}_{\text{Term (a)}} 
      +\underbrace{\|\mathcal{T}[n^{(k)}] - \mathcal{T}[n^*]\|}_{\text{Term (b)}}. 
  \label{eq:ekExpand}
\end{align}

\textit{Step 2:} We bound Term (b) by contractivity.
By Assumption~\ref{assump:contract}, for $n_1,n_2\in U$,
\[
  \|\mathcal{T}[n_1]-\mathcal{T}[n_2]\|\;\le\;\kappa\,\|n_1 - n_2\|.
\]
Since $n^{(k)}$ and $n^*$ lie in $U$, we get
\[
  \|\mathcal{T}[n^{(k)}]-\mathcal{T}[n^*]\| \;\le\;\kappa\,\|\,n^{(k)}-n^*\|\;=\;\kappa\,e_k.
\]

\textit{Step 3:} We bound Term (a) by uniform XC error.
From Sec.~\ref{sec:proofs}, we recall that $\widetilde{\mathcal{T}}$ and $\mathcal{T}$ differ only in the XC potential part, and by Assumption~\ref{assump:uniform}, that potential difference is at most $K\,\varepsilon$. 
To convince yourself, see the KS equations in Sec.~\ref{sec:theory}, where Hamiltonian operator $\hat H[n]=-\nabla^2 / 2+v_\mathrm{KS}[n]
$ ultimately defines the SCF map $\mathcal{T}$ through XC potential $v_{\mathrm{XC}}$.
In the local region $U$, changing the potential by at most $K\,\varepsilon$ changes the resulting density by at most some constant factor $\alpha\,\varepsilon$. 
Hence we have
\begin{equation}\label{eq:TTtildeDiff}
  \|\widetilde{\mathcal{T}}[n] - \mathcal{T}[n]\|\;\le\;\alpha\,\varepsilon
  \quad\text{for all}\;n\in U.
\end{equation}

Thus, for $n^{(k)}\in U$,
\[
  \|\widetilde{\mathcal{T}}[n^{(k)}] - \mathcal{T}[n^{(k)}]\|
  \;\le\;\alpha\,\varepsilon.
\]

\textit{Step 4:} We combine the previous statements into Eq.~\eqref{eq:ekExpand}:
\begin{align}
  e_{k+1}
  &\;=\;\|n^{(k+1)}-n^*\| \nonumber\\
  &\;\le\;\|\widetilde{\mathcal{T}}[n^{(k)}]-\mathcal{T}[n^{(k)}]\|
       + \|\mathcal{T}[n^{(k)}]-\mathcal{T}[n^*]\|\nonumber\\
  &\;\le\;\alpha\,\varepsilon \;+\;\kappa\,e_k.\label{eq:eIter}
\end{align}

\textit{Step 5:} The density $n^{(k)}$ remains in $U$ by Assumption~\ref{assump:approxIter}, we solve the recurrence by rewriting
\[
  e_{k+1} \;\le\;\kappa\,e_k \;+\;\alpha\,\varepsilon,
\]
as 
\[
  e_k \;\le\;\kappa^k\,e_0 + \alpha\,\varepsilon\sum_{j=0}^{k-1}\kappa^j
  \;=\;\kappa^k\,e_0 + \alpha\,\varepsilon\,\frac{\,1-\kappa^k\,}{\,1-\kappa\,}.
\]
Set $n^{(0)}$ as the initial guess, so $e_0 = \|n^{(0)}-n^*\|$. This proves the finite-$k$ bound \eqref{eq:e_kBound}. Taking $k\to\infty$, we get
\begin{equation}\label{eq:finalLimit}
  e_\infty = \|\,\widetilde{n} - n^*\| \;\le\;\frac{\alpha\,\varepsilon}{\,1-\kappa\,}.
\end{equation}
Hence the approximate self-consistent density $\widetilde{n}$ is within $\mathcal{O}(\varepsilon)$ of the true electron density $n^*$.

\textit{Step 6:} Energy difference is also $O(\varepsilon)$.
Finally, define $\widetilde{E}\equiv \widetilde{E}[\widetilde{n}]$ and $E^*\equiv E[n^*]$. Write 
\begin{align}
 \widetilde{E} - E^*
  \;&=\;\bigl(\widetilde{E}[\widetilde{n}]-E[\widetilde{n}]\bigr)\;+\;\bigl(E[\widetilde{n}]-E[n^*]\bigr).
\end{align}
Since $\widetilde{E}_{\mathrm{XC}}$ differs from $E_{\mathrm{XC}}$ by at most $\varepsilon$ in value, $\widetilde{E}$ differs from $E$ by at most $\varepsilon$ at the \emph{same} density $\widetilde{n}$. 
Meanwhile, we assume that $|E[\widetilde{n}]-E[n^*]|\le M\,\|\widetilde{n}-n^*\|$ for some Lipschitz constant $M > 0$ near $n^*$ since according to Eq.~\eqref{eq:finalLimit}, $\tilde{n}$ and $n^{\ast}$ are order $\varepsilon$ close in the region $U$ by Assumption~\ref{assump:approxIter}. 
Hence,
\[
  |\widetilde{E}-E^*|
  \le \varepsilon + M\,\|\widetilde{n}-n^*\|
  \le \varepsilon + M\,\frac{\alpha}{1-\kappa}\,\varepsilon
  = C'\,\varepsilon,
\]
where $C' = 1 + \frac{M\,\alpha}{1-\kappa}$. This completes the proof.
\end{proof}

Thus, the final approximate energy is $\mathcal{O}(\varepsilon)$-close to the true energy, demonstrating that, under our assumptions, catastrophic error accumulation in the SCF loop is avoided.

\section{Lipschitz continuity of QNNs}\label{sec:app-lipschitz}

In this section, we establish the Lipschitz continuity of QNNs employed in this work, following the derivations and notation from ref.~\cite{schuld2019evaluating} for clarity.
Recall the definition of quantum models, Eq.~\eqref{eq:qmodel}, here denoted as
$$
f_{\theta}(n) = \langle 0 | U^{\dagger}(n, \theta) \, \hat{C} \, U(n, \theta) | 0 \rangle,
$$
which forms the basis of all our $N_q$-qubit models.
We propose the following Lemma to support our arguments from Sec.~\ref{sec:proofs}.

\begin{lemma}[QNN Lipschitz continuity]\label{lem:qnn-lipschitz-L1}
Let QNN be represented by $f_{\theta}(n)=\langle 0 \mid U^{\dagger}(n,\theta)\,\hat{C}\,U(n,\theta)\mid 0\rangle$ with
\begin{align}
U(n)&=W^{(L+1)} S(n) W^{(L)} \ldots W^{(2)} S(n) W^{(1)}, \nonumber \\
S(n)&=\;e^{-i n H}, \  \ \ \hat{C}=\sum_{i=1}^{N_q}\hat{Z}_i,  \nonumber
\end{align}

where $W^k \equiv W^{k}(\theta)$ are $x$-independent unitaries (representing ansatz), $S(n)$ is the feature map with Hermitian operator $H$ with eigenvalues $\{\lambda_m\}_{m=1}^{d}$, $d = 2^{N_q}$, and $L=N_{l}$ is the number of layers. 
Then $f_{\theta}$ is globally Lipschitz in $x$, defined according to ref.~\cite{schuld2019evaluating} as
\[
f(n)=\sum_{\boldsymbol{k}, \boldsymbol{j} \in[d]^L} e^{i\left(\Lambda_{\boldsymbol{k}}-\Lambda_{\boldsymbol{j}}\right) n} a_{\boldsymbol{k}, \boldsymbol{j}},\qquad
\Lambda_j=\lambda_{j_1}+\cdots+\lambda_{j_L},
\]

Consequently,
\[
\sup_{x\in\mathbb{R}}|f'_{\theta}(n)|
\le N_q \sum_{\boldsymbol{k}, \boldsymbol{j} \in[d]^L}|\Lambda_{\boldsymbol{k}}-\Lambda_{\boldsymbol{j}}|<\infty,
\]
where we use multi-index $\boldsymbol{j}=$ $\left\{j_1, \ldots, j_L\right\} \in[d]^L$, $[d]^L$ denotes any set of $L$ integers between $1$ and $d$, and $N_q$ is the number of qubits.
Hence, $f_{\theta}$ is globally Lipschitz on $\mathbb{R}$ (for any density $n$) with a $\theta$-independent constant. 
\end{lemma}

\begin{proof}

We follow the ref.~\cite{schuld2019evaluating} and recall that:
\[
f(n)=\sum_{\boldsymbol{k}, \boldsymbol{j} \in[d]^L} e^{i\left(\Lambda_{\boldsymbol{k}}-\Lambda_{\boldsymbol{j}}\right) n} a_{\boldsymbol{k}, \boldsymbol{j}},\qquad
\]

where 

\begin{align}
a_{\boldsymbol{k}, \boldsymbol{j}}=\sum_{i, i^{\prime}}\left(W^*\right)_{1 k_1}^{(1)}\left(W^*\right)_{j_1 j_2}^{(2)} \ldots\left(W^*\right)_{j_L i}^{(L+1)} C_{i, i^{\prime}} \nonumber \\
\times W_{i^{\prime} j_L}^{(L+1)} \ldots W_{j_2 j_1}^{(2)} W_{j_1 1}^{(1)}. \nonumber
\end{align}

Since $\|\hat{C}\|=N_q$, $W^{k}$ are unitary, and $\sup_{n}|f_\theta (n)| \leq N_q$, consequently the elements $|a_{\boldsymbol{k}, \boldsymbol{j}}|\le \|\hat{C}\|= N_q$.
We can then write the derivative of our model, 

\[
f_{\theta}'(n)=\sum_{\boldsymbol{k}, \boldsymbol{j} \in[d]^L} i \left(\Lambda_{\boldsymbol{k}}-\Lambda_{\boldsymbol{j}} \right)e^{i\left(\Lambda_{\boldsymbol{k}}-\Lambda_{\boldsymbol{j}}\right) n} a_{\boldsymbol{k}, \boldsymbol{j}}.
\]

So $\sup_n|f'_{\theta}(n)|\le \sum_{\boldsymbol{k}, \boldsymbol{j} \in[d]^L}|\Lambda_{\boldsymbol{k}}-\Lambda_{\boldsymbol{j}}|\,|a_{\boldsymbol{k}, \boldsymbol{j}}|\le N_q\sum_{\boldsymbol{k}, \boldsymbol{j} \in[d]^L}|\Lambda_{\boldsymbol{k}}-\Lambda_{\boldsymbol{j}}|<\infty$ since the spectrum is finite. 
The mean‑value theorem implies the stated Lipschitz continuity of $f_{\theta}(n)$. 
\end{proof}

\bibliographystyle{apsrev4-1}
\bibliography{bibliography}

\end{document}